\documentclass[12pt]{article}
\usepackage{amsfonts}
\usepackage{amsmath}
\usepackage{amssymb}
\usepackage{amsthm}
\usepackage{amssymb}
\usepackage[all]{xy}
\usepackage{xcolor}
\newif\ifshowrevisions
\showrevisionstrue
\definecolor{revisionpurple}{RGB}{140,55,180}
\definecolor{revisiongreen}{RGB}{0,122,78}
\ifshowrevisions\else
\colorlet{revisionpurple}{black}
\colorlet{revisiongreen}{black}
\fi
\usepackage{authblk}
\usepackage{graphicx}
\usepackage{geometry}
\usepackage{enumerate}
\usepackage{subcaption}
\usepackage[colorlinks=true, citecolor=blue, linkcolor=blue, hyperfootnotes=false, hypertexnames=false]{hyperref}

\usepackage{float}
\usepackage{setspace}
\usepackage{bbm}
\usepackage{bm}
\usepackage{booktabs}
\usepackage{tabu}
\usepackage[utf8]{inputenc}
\usepackage{tikz}
\usepackage{multicol}
\usepackage{multirow}
\usepackage{natbib}
\usepackage{pdflscape}
\usepackage[normalem]{ulem}
\usepackage{appendix}
\usepackage{placeins}
\usepackage{etoolbox}
\AtBeginEnvironment{quote}{\par\singlespacing\small}

\newcommand{\cube}{\{-1,1\}^n}

\newcommand{\Dn}{\mathcal D_n}

\theoremstyle{plain}
\newtheorem{theorem}{Theorem}
\newtheorem{corollary}{Corollary}
\newtheorem{lemma}{Lemma}

\newtheorem{proposition}[theorem]{Proposition}

\newtheorem{axiom}{Axiom}

\theoremstyle{definition}

\newtheorem{example}[theorem]{Example}

\title{Model Complexity and Restrictiveness}
\author{Keaton Ellis and Sara Neff\thanks{Ellis: Department of Economics, Monash University (\href{mailto:keaton.ellis@monash.edu}{keaton.ellis@monash.edu}); Neff: University of California, Berkeley (\href{mailto:sara_neff@berkeley.edu}{sara\_neff@berkeley.edu}). We are grateful to Federico Echenique, Michael Jansson, Shachar Kariv, Demian Pouzo, Chris Shannon, and seminar participants of the UC Berkeley Theory Lunch, 2025 California Econometrics Conference, and ESIF 2026 for helpful comments. Keaton acknowledges and is  grateful for support from the Foundations of Data Science Institute (FODSI), funded by the National Science Foundation TRIPODS program, and for the hospitality of the Simons Institute for the Theory of Computing at the University of California, Berkeley.} \\
Latest draft \href{https://drive.google.com/file/d/1dzmsolp9OFVlFpjzJK9b6smO_bfSQLvY/view?usp=sharing}{here}}

\begin{document}

\maketitle
\begin{abstract}
    We study the measure of restrictiveness proposed in \cite{fudenberg2023restrictiveness} to evaluate complexity of economic models using synthetic data. We show that Rademacher complexity is an affine transformation of a particular case of a consistent finite-sample estimate of restrictiveness. Our results show that restrictiveness inherits a cardinal interpretation as a bound on generalization error while avoiding the inability of limiting Rademacher complexity to distinguish between some falsifiable models.
\end{abstract}

\noindent\textbf{Keywords:} complexity, model, restrictiveness.

\newpage

\section{Introduction}\label{sec:introduction}

A model's fit is more informative when it rules out plausible behavior. Measures of model complexity help assess how much of a good fit reflects a model's flexibility. The economic theory literature has recently proposed \textit{restrictiveness} (\citealp{fudenberg2023restrictiveness}), a measure describing the flexibility of a model by evaluating its expected distance (according to some discrepancy function) to synthetic data from a pre-specified data-generating process, whose support is referred to as the ``eligible set''. A completely unrestrictive model is one which always achieves within sample loss of zero on the synthetic data. Both the discrepancy function and the eligible set are free choices by the analyst, making restrictiveness quite a general measure applicable across economic domains. This measure is viewed as complementary to measures of predictive quality estimated on real data, such as \textit{completeness} (\citealp{fudenberg2022measuring}).
In this framework, the ideal model is one that is fully complete and fully restrictive - it captures the true structure of the data well via properly-specified structure \citep{liang2026using}. 
In this paper, we investigate how a measure such as restrictiveness relates to complexity measures from computer science.

We connect restrictiveness to Rademacher complexity, a measure from statistical learning theory with deep ties to falsifiability. Rademacher complexity asks how well a function class can fit pure noise: given a sample of $n$ model covariates, it measures the maximal correlation the class can achieve with an independently drawn vector of uniform $\{-1,1\}$ outputs. In practice, Rademacher complexity bounds \textit{generalization error} \textemdash the maximum error when estimating the mean of any function using a sample (\citealp{mitzenmacher2017probability}). Our main result is an equivalence: when the synthetic data is approximately iid Rademacher and discrepancy is measured by a correlation-based loss function, then for every $n$, the gap between Rademacher complexity and (one minus) an empirical analog of restrictiveness can be made arbitrarily small. These results show that restrictiveness contains classic measures of model flexibility as a foundation, and subsequently generalizes them through a free choice of (i) the discrepancy function and (ii) the data-generating process, as opposed to using Rademacher random noise. Furthermore, we show that the axiomatization of restrictiveness in \cite{fudenberg2023restrictiveness} can be easily augmented to provide an axiomatization of Rademacher complexity.

Our main results imply that, since Rademacher complexity bounds generalization error, in some cases a finite sample version of restrictiveness can be converted into a bound on a model's generalization error. More broadly, restrictiveness inherits a similar type of cardinal interpretation while avoiding limitations of Rademacher complexity. As $n \to \infty$, the Rademacher complexity of a binary-valued function class vanishes if the class has finite VC dimension, that is, if the class rules out at least one outcome pattern on every sufficiently large set of distinct covariates \citep[Lemma~4 and Theorem~6]{bartlett2002rademacher}. Thus, under correlation discrepancy and increasingly fine versions of our eligible set, every such model becomes fully restrictive \textemdash in the limit, this environment cannot discriminate among falsifiable models at all. This limitation is meaningful for economic analysis. For example, in the finite-state binary-choice setting, expected utility theory (EUT) and some strict extensions, such as Choquet expected utility (CEU, \citealp{schmeidler1989subjective}), are both falsifiable with finite data, thus stating that some extensions of EUT possess the same level of flexibility \citep{basu2020falsifiability}. While there is a difference in ``how much'' finite data is needed to distinguish EUT and CEU, this distinction applies only in finite samples and not in the limiting case.


The limitation, however, is particular to the eligible set of iid Rademacher random variables. Continuing the example, an eligible prior whose support is contained in CEU and strictly contains EUT generates a wedge between the two models. Thus, we find that the particular importance of restrictiveness arises from its ability to set economically-meaningful eligible sets. At the same time, doing so does not automatically retain the same interpretation for bounding generalization error.

Despite this, we identify sufficient conditions under which restrictiveness retains a useful cardinal interpretation in economically meaningful environments. In particular, with convexity and boundedness assumptions on the discrepancy, and monotonicity assumptions on the model and eligible set of functions, an affine transformation of finite-sample restrictiveness bounds the generalization error of a discrepancy function, i.e. the expected largest difference between population discrepancy and sample discrepancy across the model class. These types of monotonicity conditions on model behavior are natural in many economically relevant settings.
Overall, we find that restrictiveness provides a practical generalization of Rademacher complexity to economic contexts while retaining a cardinal connection to generalization error concepts under natural conditions.

The formal analysis makes these connections by relating the population and finite-sample objects. While restrictiveness is a population object that measures how well a function class can match functions from an eligible set, Rademacher complexity is a finite sample object measuring correlation to a vector of Rademacher random variables that explicitly depends on a sample of size $n$. To connect these measures, we first define empirical analogs of the objects that appear in restrictiveness.
We show that a connection can be made between these objects in two ways. Proposition \ref{prop:equivalence_in_sample} says that there is an equivalence between this empirical analog of restrictiveness and empirical Rademacher complexity when the sample is known. Proposition \ref{prop:bounded_difference} says that if only the sample \textit{size} $n$ is known, the expected difference (over samples of size $n$) between the measures
can be made arbitrarily small. {Propositions \ref{prop:offset_equivalence_in_sample} and \ref{prop:offset_bounded_difference} extend these connections to squared-error discrepancy and offset Rademacher complexity, which adds a quadratic penalty to the random-sign fit criterion (\citealp{liang2015learning}).} Lemma \ref{lemma:convergence} provides conditions under which our sample analog consistently estimates the restrictiveness numerator with the eligible prior fixed. We show that convergence holds for a large class of discrepancy functions, including the function analogous to Rademacher complexity (defined below), in a setting where the covariate-generating process is unknown. Note that this final result also allows for restrictiveness to be applied in a broader class of empirical settings than controlled laboratory experiments (e.g., \citealp{gentzkow2024pricing}), and provides a microfoundation for a related discussion in work by \cite{fudenberg2026model} (see below). {Finally, Proposition \ref{prop:endpt-bound} gives sufficient conditions under which an affine transformation of finite-sample restrictiveness bounds expected generalization error in discrepancy for economically meaningful eligible priors.}

Section \ref{sec:notation_and_objects} introduces our notation and explains the objects that appear in the statement of our results. Section \ref{restrad} states our main results, as well as the lemmas we use to prove them. 
Section \ref{sec:generalization_error} gives the generalization-error bound. {Section \ref{sec:interpretation} concludes.}

\section{Notation and Objects}\label{sec:notation_and_objects}

\subsection{Notation}\label{sec:notation}

Let $\mathcal{X}$ be a closed and bounded infinite subset of $\mathbb{R}^m$ for some finite $m \in \mathbb{N}$, $X$ be a random vector of $\mathcal{X}$ with distribution $P_X$, and $x \in \mathcal{X}$ be an arbitrary element of $\mathcal{X}$, potentially a draw from $X$. Similarly, let {$\mathcal{Y} = [-1, 1]$} be the outcome space, $Y$ be a random vector over $\mathcal{Y}$, and $y \in Y$ be an arbitrary element of $\mathcal{Y}$, potentially a draw from $Y$.\footnote{{Predictions may be restricted to a subset of $\mathcal Y$. For example, a $k$-ary classification problem has predictions in $\{ -1 + \frac{2(i-1)}{k-1}|i=1,\ldots,k\}$. Setting $k=3$ gives $\{-1,0,1\}$, while $k=4$ gives $\{-1,-\frac{1}{3},\frac{1}{3},1\}$. Alternatively, consider a decision maker choosing bundles $(x_1,x_2)$ along a linear budget set. Relative demand for the first good, $\frac{x_1}{x_1+x_2}$, takes values in $[0,1]\subseteq\mathcal Y$.} Note further that any bounded problem can be transformed into one in this space. Our propositions then apply after the corresponding changes of scale and normalization.} We use $S_n = (x_1, \dots, x_n)$ to denote an iid sample of covariates of arbitrary size $n$, dropping the $n$ subscript when there is no confusion.

We are interested in analyzing functions from $\mathcal{X}$ to $\mathcal{Y}$. Let $\bar{\mathcal{F}} = \{ f:\mathcal{X}\rightarrow \mathcal{Y} \}$ be the set of all such functions, and let $\mathcal{F} \subseteq \bar{\mathcal{F}}$ be an arbitrary set of functions. A parametric \textit{model} is a set of functions $\mathcal{F}_\Theta = \{ f_\theta \}_{\theta \in \Theta}$ defined by a nonempty compact metric parameter space $\Theta$. We also define a specific subset of $\bar{\mathcal{F}}$. Let $\bar{\mathcal{F}}_\sigma = \{ f:\mathcal{X} \rightarrow \{ -1, 1 \} \}$ be the set of all functions to the set $\{ -1, 1 \}$. We let $\sigma \in \bar{\mathcal{F}}_\sigma$ denote an arbitrary function mapping $\mathcal{X}$ to $\{-1, 1\}$, and $\sigma_n \in \{-1, 1\}^n$ be a vector, again dropping the $n$ subscript when there is no confusion.

Next, let $\mathcal{X}_\varepsilon$ be a coarsening of $\mathcal X$ obtained from a grid with side lengths at most $\varepsilon$. Formally, let $a_d$ and $b_d$ denote the minimum and maximum values of coordinate $d=1,\dots,m$ on $\mathcal X$. The grid is the Cartesian product
\[
\times_{d=1}^m\left\{\left[a_d,\ a_d+\varepsilon\right),\left[a_d+\varepsilon,\ a_d+2\varepsilon\right),\dots,\left[a_d+\left\lfloor\frac{b_d-a_d}{\varepsilon}\right\rfloor\varepsilon,\ b_d\right]\right\}.
\]
The intervals preceding the final interval are indexed by $j=0,\dots,\lfloor(b_d-a_d)/\varepsilon\rfloor-1$ and have endpoints $a_d+j\varepsilon$ and $a_d+(j+1)\varepsilon$. If $b_d-a_d<\varepsilon$, the list consists only of $[a_d,b_d]$; in particular, a degenerate coordinate contributes $\{a_d\}$. Intersect each grid cell with $\mathcal X$ and discard empty cells to obtain $\mathcal X_\varepsilon$.

Let $x_\varepsilon \in \mathcal{X}_\varepsilon$ denote an arbitrary element of $\mathcal{X}_\varepsilon$, $X_\varepsilon$ denote a random vector over $\mathcal{X}_\varepsilon$ with distribution $P_{X_\varepsilon}$, where $P_{X_\varepsilon}(X_\varepsilon = x_\varepsilon) = P_X(X \in x_\varepsilon)$. We also define analogous function classes over $\mathcal{X}_\varepsilon$: $\bar{\mathcal{F}}_\varepsilon = \{ f_\varepsilon: \mathcal{X}_\varepsilon \rightarrow \mathcal{Y} \}$ is the set of all functions, $\bar{\mathcal{F}}_{\sigma, \varepsilon} = \{ f_\varepsilon: \mathcal{X}_\varepsilon \rightarrow \{ -1, 1 \} \}$ is the set of all functions mapping to $\{-1, 1\}$, and $\mathcal{F}_\varepsilon$ denotes an arbitrary set of functions. We abuse notation by using $\mathcal{F}_\varepsilon$ to describe both a set of functions mapping from $\mathcal{X}_\varepsilon$ and the analogous functions in $\mathcal{X}$. Naturally, for some function $f_\varepsilon$ there exists an analogous $f$ such that if $x \in x_\varepsilon$, then $f(x) = f_\varepsilon(x_\varepsilon)$.

As final pieces of notation, let a \textit{discrepancy function} $d:\bar{\mathcal{F}} \times \bar{\mathcal{F}} \rightarrow \mathbb{R}_+$ be a function over the set of all functions such that $d(f, f') = 0$ if and only if $f(X) = f'(X)$ a.s. under $P_X$. We assume the discrepancy function is of the form \[ d(f_\theta, f) = E_{P_X}[g(X, \theta; f)], \] 

\noindent where $g(x, \theta; f)$ encompasses a large number of relevant objective functions including squared error and KL divergence. We assume (i) there exists an envelope function $G(x)$ such that $|g(x, \theta; f)| \leq G(x)$ for all $x \in \mathcal{X}$, $\theta \in \Theta$ {and $f\in\mathcal F$}, and $E[G(X)]<\infty$; and (ii) {for almost every $f$ under $\lambda_{\mathcal F}$, $g(X,\theta;f)$ is continuous in $\theta$ with probability one}.

Given a sample $S$, we define an empirical analog to $d$, $d_n$, defined as \[ d_n(f_\theta, f) = \frac{1}{n}\sum_{i=1}^n g(x_i, \theta; f) \]

Let $\lambda_\mathcal{F}$ be the uniform distribution over the set $\mathcal{F}$. We restrict attention to nonempty eligible sets for which this uniform distribution is well defined. For finite $\mathcal F$, each function receives probability $1/|\mathcal F|$. We assume that model predictions are jointly measurable in $(x,\theta)$, eligible predictions in $(x,f)$, and losses in $(x,\theta,f)$.

\subsection{Objects}\label{sec:objects}

We introduce eight objects that we work with to obtain our results - four related to Rademacher complexity, and four related to restrictiveness. For quick reference, the measures are summarized in Table \ref{tab:objects}.

\begin{table}
    \centering
    \resizebox{\textwidth}{!}{%
    \begin{tabular}{llr}
        Name & Notation & Expression \\ \midrule
        Restrictiveness & $r(\mathcal{F}_{\Theta},\mathcal{F},d)$ & $\frac{\mathbb{E}_{\lambda_\mathcal{F}}[\text{ inf}_{f_\theta \in \mathcal{F}_{\Theta}}\ d(f_\theta, f)]}{\mathbb{E}_{\lambda_\mathcal{F}}\big[d(f_{\text{base}},f)\big]}$ \\[5mm]
        
        Approximation error & $e(\mathcal{F}_{\Theta},\mathcal{F},d)$ & $\mathbb{E}_{\lambda_\mathcal{F}} \Big[\underset{f_\theta \in \mathcal{F}_{\Theta}}{\text{ inf}}\ d(f_\theta, f)\Big]$ \\[5mm]
        
        Empirical approximation error & $\tilde{e}_n(\mathcal{F}_{\Theta},\mathcal{F},d_n,S)$ & $\mathbb{E}_{\lambda_{\mathcal{F}}} \Big[\underset{f_\theta \in \mathcal{F}_{\Theta}}{\text{inf}}\ d_n(f_\theta,f)\Big]$ \\[5mm]

        Average approximation error & $e_n(\mathcal{F}_{\Theta},\mathcal{F},d_n)$ & $\mathbb{E}_S \mathbb{E}_{\lambda_{\mathcal{F}}} \Big[\underset{f_\theta \in \mathcal{F}_{\Theta}}{\text{inf}}\ d_n(f_\theta,f)\Big]$ \\[5mm]

        Empirical RC & $\tilde{R}_n(\mathcal{F}_\Theta,S)$ & $\mathbb{E}_\sigma \Bigg[ \underset{f_\theta \in \mathcal{F}_{\Theta}}{\text{sup}}\ \frac{1}{n} \sum^n_{i=1} \sigma_i \cdot f_\theta(x_i) \Bigg]$ \\[5mm]

        Rademacher complexity (RC) & $R_n(\mathcal{F}_{\Theta})$ & $\mathbb{E}_S \mathbb{E}_\sigma \Bigg[ \underset{f_\theta \in \mathcal{F}_{\Theta}}{\text{sup}}\ \frac{1}{n} \sum^n_{i=1} \sigma_i \cdot f_\theta(x_i) \Bigg]$ \\[5mm]

        Empirical Offset RC & $\tilde{R}^{\text{off}}_n(\mathcal{F}_{\Theta},S,c)$ & $ \mathbb{E}_\sigma \Bigg[ \underset{f_\theta \in \mathcal{F}_{\Theta}}{\text{sup}}\ \frac{1}{n} \sum^n_{i=1} \left[\sigma_i \cdot f_\theta(x_i) - c f_\theta(x_i)^2\right] \Bigg]$ \\[5mm]

        Offset RC & $R^{\text{off}}_n(\mathcal{F}_{\Theta},c)$ & $\mathbb{E}_S \mathbb{E}_\sigma \Bigg[ \underset{f_\theta \in \mathcal{F}_{\Theta}}{\text{sup}}\ \frac{1}{n} \sum^n_{i=1} \left[\sigma_i \cdot f_\theta(x_i) - c f_\theta(x_i)^2\right]  \Bigg]$ \\[5mm]


    \end{tabular}
    }
    \caption{A summary of objects, where RC stands for Rademacher complexity}
    \label{tab:objects}
\end{table}

\paragraph{Rademacher complexity}

Following \cite{mitzenmacher2017probability}, for a given sample of covariates $S$, \textit{empirical Rademacher complexity} is defined as \[\tilde{R}_n(\mathcal{F}_\Theta,S) = \mathbb{E}_{\sigma} \left[ \underset{f_\theta \in \mathcal{F}_{\Theta}}{\text{sup}}\ \frac{1}{n} \sum^n_{i=1} \sigma_i \cdot f_\theta(x_i) \right].\]

\noindent That is, it calculates the average maximum correlation with uniform random noise achievable by class $\mathcal{F}_\Theta$ on a sample $S$; the average is taken over the vector $\sigma = (\sigma_1, \dots \sigma_n)$ where each $\sigma_i$ is an independent and identically distributed Rademacher random variable, i.e. takes values $1$ and $-1$ with equal probability. \textit{Rademacher complexity} is simply the expectation over samples $S$ with respect to $P_X$, defined as \[R_n(\mathcal{F}_{\Theta}) = \mathbb{E}_S \mathbb{E}_\sigma \left[ \underset{f_\theta \in \mathcal{F}_{\Theta}}{\text{sup}}\ \frac{1}{n} \sum^n_{i=1} \sigma_i \cdot f_\theta(x_i) \right].\]

Both measures depend on the sample size $n$: $R_n$ averages over samples drawn from $P_X$, while empirical Rademacher complexity additionally depends on the realized sample $S$. They control uniform differences between sample averages and population expectations, yielding bounds on convergence rates and generalization error (\citealp{bartlett2002rademacher,bartlett2005local}). For prediction risk, these bounds are applied to the corresponding loss class. Rademacher penalties also provide a criterion for model selection by balancing empirical fit against the complexity of the model class (\citealp{lozano2000model,bartlett2002model}).

Offset Rademacher complexity is introduced in \cite{liang2015learning} as a notion related to Rademacher complexity which is appropriate for learning with square loss. It differs from Rademacher complexity in that it includes a negative quadratic term for some $c \geq 0$. Define \textit{empirical offset Rademacher complexity} as \[\tilde{R}^{\text{off}}_n(\mathcal{F}_\Theta,S,c) = \mathbb{E}_{\sigma} \left[ \underset{f_\theta \in \mathcal{F}_{\Theta}}{\text{sup}}\ \frac{1}{n} \sum^n_{i=1} \left[\sigma_i \cdot f_\theta(x_i) - c f_\theta(x_i)^2\right] \right],\]

\noindent and define \textit{offset Rademacher complexity} as \[R^{\text{off}}_n(\mathcal{F}_\Theta,c) = \mathbb{E}_{S} \mathbb{E}_{\sigma} \left[ \underset{f_\theta \in \mathcal{F}_{\Theta}}{\text{sup}}\ \frac{1}{n} \sum^n_{i=1} \left[\sigma_i \cdot f_\theta(x_i) - c f_\theta(x_i)^2\right] \right].\]

\paragraph{Restrictiveness} As defined in \cite{fudenberg2023restrictiveness}, let \[e(\mathcal{F}_{\Theta},\mathcal{F},d) = \mathbb{E}_{\lambda_\mathcal{F}}[\text{ inf}_{f_\theta \in \mathcal{F}_{\Theta}}\ d(f_\theta, f)]\] denote the \textit{approximation error} of a function class $\mathcal{F}_\Theta$ relative to some arbitrary function class $\mathcal{F}$ and discrepancy function $d$. \textit{Restrictiveness} scales this measure by the performance of a baseline function $f_{base} \in \mathcal{F}_\Theta$:

\[ r(\mathcal{F}_{\Theta},\mathcal{F},d) = \frac{e(\mathcal{F}_{\Theta},\mathcal{F},d)}{e(f_{base},\mathcal{F},d)} = \frac{\mathbb{E}_{\lambda_\mathcal{F}}[\text{ inf}_{f_\theta \in \mathcal{F}_{\Theta}}\ d(f_\theta, f)]}{\mathbb{E}_{\lambda_\mathcal{F}}\big[d(f_{\text{base}},f)\big]}. \]

\noindent As in Assumption 1 of \cite{fudenberg2023restrictiveness}, assume the baseline denominator is positive. Note that since $f_{base} \in \mathcal{F}_\Theta$ and $d(f, g) \geq 0$ for all functions $f, g \in \bar{\mathcal{F}}$, it must be the case that $r(\mathcal{F}_{\Theta},\mathcal{F},d) \in [0, 1]$.

Since Rademacher complexity is indexed by sample size $n$ and restrictiveness is a population measure, we also define two empirical analogs of approximation error to make meaningful comparisons. Given a sample $S$ of size $n$, let the \textit{empirical (approximation) error} be $\tilde{e}_n(\mathcal{F}_{\Theta},\mathcal{F},d_n, S) = \mathbb{E}_{\lambda_{\mathcal{F}}} \left[\underset{f_\theta \in \mathcal{F}_{\Theta}}{\text{inf}}\ d_n(f_\theta,f)\right]$, where $d_n$ is the empirical analog of $d$. Similar to Rademacher complexity, let the \textit{average (approximation) error} be the expected empirical (approximation) error over all possible samples $S$, defined as \[e_n(\mathcal{F}_{\Theta},\mathcal{F},d_n) = \mathbb{E}_S \mathbb{E}_{\lambda_{\mathcal{F}}} \left[\underset{f_\theta \in \mathcal{F}_{\Theta}}{\text{inf}}\ d_n(f_\theta,f)\right].\]

\section{Restrictiveness and Rademacher Complexity}\label{restrad}

Our main results relating restrictiveness and Rademacher complexity are contained in four propositions. Proposition \ref{prop:equivalence_in_sample} states that, for a sample of distinct covariates, when the discrepancy is one minus correlation and the eligible set consists of all functions mapping a sufficiently fine $\mathcal X_\varepsilon$ to $\{-1,1\}$, empirical Rademacher complexity equals one minus empirical approximation error. Building on this connection, Appendix~\ref{sec:axiomatization} augments the axioms of restrictiveness to characterize empirical Rademacher complexity. Under atomlessness, Proposition \ref{prop:bounded_difference} states that, for each sample size, the corresponding difference in expectation can be made arbitrarily small by choosing a sufficiently fine eligible partition. Propositions \ref{prop:offset_equivalence_in_sample} and \ref{prop:offset_bounded_difference} produce similar connections with squared error and offset Rademacher complexity. Lemma \ref{lemma:convergence} completes the comparison: for a fixed finite eligible set, average approximation error converges to the numerator of restrictiveness under our maintained assumptions.

\subsection{Main Results}\label{mainresults}

Lemma \ref{lemma:sigma_function} establishes an equivalence between empirical Rademacher complexity (the left hand side of Equation \ref{eq-lemma1}) and an object that replaces the Rademacher random variables with a set of functions that approximately map from the sample $S$ to $\{-1,1\}$. 

\begin{lemma}\label{lemma:sigma_function}
    Given a sample of unique covariates $S = (x_1, \dots, x_n)$, there exists a function class $\bar{\mathcal{F}}_{\sigma, \varepsilon}$ such that \begin{align}\label{eq-lemma1}\mathbb{E}_\sigma \Bigg[ \underset{f_\theta \in \mathcal{F}_{\Theta}}{\sup}\ \frac{1}{n} \sum^n_{i=1} \sigma_i \cdot f_\theta(x_i) \Bigg] = \mathbb{E}_{\lambda_{\bar{\mathcal{F}}_{\sigma, \varepsilon}}} \Bigg[ \underset{f_\theta \in \mathcal{F}_{\Theta}}{\sup}\ \frac{1}{n} \sum^n_{i=1} f(x_i) \cdot f_\theta(x_i) \Bigg],\end{align}

\end{lemma}

\begin{proof}
    Note that incongruencies between the two arise only when there exist some $i \neq j$ such that $x_i$ and $x_j$ belong to the same partition, since a function representation cannot accommodate the flexibility of $\sigma_i \neq \sigma_j$, which is permitted on the left-hand side. If one can guarantee that $x_i$ and $x_j$ are not in the same partition for all $i \neq j$, then the equality is immediate. 
    
    For $n=1$, any partition suffices. For $n\ge2$, choosing $\varepsilon < \min_{i\neq j}||x_i-x_j||_\infty$ guarantees that the interval coarsening will be fine enough to place $x_i$ into a unique partition.
\end{proof}

Note that the uniform distribution $\lambda_{\bar{\mathcal{F}}_{\sigma,\varepsilon}}$ is well defined since $\bar{\mathcal{F}}_{\sigma,\varepsilon}$ is a finite set. In the proof, $\varepsilon$ is chosen such that the set of functions $\bar{\mathcal{F}}_{\sigma,\varepsilon}$ can be interpreted as containing all of the empirical content of all functions mapping from $\mathcal{X}$ to $\{-1,1\}$ when observing a sample $S$.

Below, Lemma \ref{lemma:inf_sup_correlation} says that by choosing a particular form of the empirical discrepancy $d_n$, the object on the right hand side of Equation \ref{eq-lemma1} equals one minus the minimum of $d_n$ over the set of functions $\mathcal{F}_{\Theta}$. Together, the two lemmas lead to Proposition \ref{prop:equivalence_in_sample}, which establishes an equivalence between empirical Rademacher complexity and empirical approximation error. 

\begin{lemma}\label{lemma:inf_sup_correlation}
    Fix a sample $S=(x_1,\ldots,x_n)$, a mapping $\sigma \in \bar{\mathcal{F}}_\sigma$, and assume a discrepancy function of the form $d_n:\bar{\mathcal{F}} \times \bar{\mathcal{F}}_\sigma \rightarrow \mathbb{R}$ such that $$d_n(f, \sigma) = 1 - \frac{\sum_{i=1}^nf(x_i)\sigma(x_i)}{n}. $$

    Then,
    
    $$\sup_{f_\theta \in \mathcal{F}_\Theta} \frac{\sum_{i=1}^nf_\theta(x_i)\sigma(x_i)}{n} = 1 - \inf_{f_\theta \in \mathcal{F}_\Theta} d_n(f_\theta, \mathbf{\sigma}).$$
\end{lemma}

\begin{proof}
    \begin{align*}
        \sup_{f_\theta \in \mathcal{F}_\Theta} \frac{\sum_{i=1}^nf_\theta(x_i)\sigma(x_i)}{n}
        & = - \inf_{f_\theta \in \mathcal{F}_\Theta} -\frac{\sum_{i=1}^nf_\theta(x_i)\sigma(x_i)}{n}
        \\
        & = 1 - \inf_{f_\theta \in \mathcal{F}_\Theta} \left(1 -\frac{\sum_{i=1}^nf_\theta(x_i)\sigma(x_i)}{n}\right)
        \\
        & = 1 - \inf_{f_\theta \in \mathcal{F}_\Theta} d_n(f_\theta, \mathbf{\sigma})
    \end{align*}
\end{proof}

Note that $d_n(f,\sigma)\geq 0$, since \[||(\sigma(x_1),\ldots,\sigma(x_n))||=\sqrt{n}\quad\text{and}\quad||(f(x_1),\ldots,f(x_n))||\leq\sqrt{n}.\] Also note that $d_n(f,\sigma)=0$ if and only if $f=\sigma$ when restricted to the sample $S$. Hence, the function $d_n$ satisfies coincidence, the singular requirement of discrepancy functions for restrictiveness (\citealp{fudenberg2023restrictiveness}).

Note also that the restriction of the second argument to $\bar{\mathcal{F}}_{\sigma}$ is important, since otherwise the second observation does not hold. For example, consider a $k$-ary classification problem with $k=3$ and $f(x)=0$. Then $d_n(f,f)=1$. An important sub-case is binary classification: if $k=2$, {both arguments are restricted to $\bar{\mathcal F}_\sigma$, and $d_n$ becomes a true discrepancy function on this class.}

\begin{proposition}\label{prop:equivalence_in_sample}
    Assume a discrepancy function $d_n$ as in Lemma \ref{lemma:inf_sup_correlation}. Given a sample of unique covariates $S = (x_1, \dots, x_n)$, there exists a function class $\bar{\mathcal{F}}_{\sigma, \varepsilon}$ such that \[ \tilde{R}_n(\mathcal{F}_\Theta, S) = 1 - \tilde{e}_n(\mathcal{F}_{\Theta}, \bar{\mathcal{F}}_{\sigma, \varepsilon}, d_n, S). \]
\end{proposition}

\begin{proof}
    The proof follows as a consequence of the two lemmas. Again setting $\varepsilon < \min_{i\neq j}||x_i-x_j||_\infty$ for $n\ge2$ (and using any partition for $n=1$):

    \begin{align*} 
        \tilde{R}_n(\mathcal{F}_\Theta, S)
        & = \mathbb{E}_\sigma \Bigg[ \underset{f_\theta \in \mathcal{F}_{\Theta}}{\sup}\ \frac{1}{n} \sum^n_{i=1} \sigma_i \cdot f_\theta(x_i) \Bigg]
        \\
        & = \mathbb{E}_{\lambda_{\bar{\mathcal{F}}_{\sigma, \varepsilon}}} \Bigg[ \underset{f_\theta \in \mathcal{F}_{\Theta}}{\sup}\ \frac{1}{n} \sum^n_{i=1} f(x_i) \cdot f_\theta(x_i) \Bigg] 
        \\
        & = \mathbb{E}_{\lambda_{\bar{\mathcal{F}}_{\sigma, \varepsilon}}} \Bigg[ 1 - \inf_{f_\theta \in \mathcal{F}_\Theta} d_n(f_\theta, f) \Bigg] 
        \\
        & = 1 - \mathbb{E}_{\lambda_{\bar{\mathcal{F}}_{\sigma, \varepsilon}}} \Bigg[ \inf_{f_\theta \in \mathcal{F}_\Theta} d_n(f_\theta, f) \Bigg]
        \\
        & = 1 - \tilde{e}_n(\mathcal{F}_\Theta, \bar{\mathcal{F}}_{\sigma, \varepsilon}, d_n, S)
    \end{align*}
\end{proof}

Proposition \ref{prop:equivalence_in_sample} establishes an equivalence between empirical Rademacher complexity and a sample analog of the numerator in restrictiveness.\footnote{Given Proposition \ref{prop:equivalence_in_sample}, we can also provide an axiomatic characterization of empirical Rademacher complexity in the spirit of Proposition 1 of \cite{fudenberg2023restrictiveness}, adapting their axioms to a finite environment and supplementing with axioms that pin down the necessary discrepancy function and eligible set. The results are shown in Appendix \ref{sec:axiomatization}.} It relies on a particular coarsening $\bar{\mathcal{F}}_{\sigma, \varepsilon}$ of the eligible set of functions. A version of Proposition \ref{prop:equivalence_in_sample} would also be true if the eligible set of functions were all functions that map from the sample to $\{-1,1\}$. However, restrictiveness is a population object, so to get equivalence between Rademacher complexity and the numerator of restrictiveness would require taking $n$ to $\infty$ - {which raises measurability concerns when $\mathcal X$ is uncountable: the unrestricted eligible set can contain nonmeasurable functions, for which population discrepancy is not defined.}

While \cite{fudenberg2023restrictiveness} do not encounter this measurability concern because they assume $\mathcal{X}$ is a finite set, we work with infinite $\mathcal{X}$ for two reasons. First, we need each element of the sample to be unique for Lemma \ref{lemma:sigma_function} to hold, which could not be true if $X$ can only take on finite values and we were to send $n$ to $\infty$. Second, we want to generalize the setting to infinite $\mathcal{X}$ for interpretability reasons. In many economically-relevant environments, we are interested in applying a model to a continuum of responses, which limits the applicability of statements on finite $\mathcal{X}$. For example, the analysis in Section 6 of \cite{fudenberg2023restrictiveness} describes the restrictiveness of various models of choice under risk, but only on a set of 25 binary lotteries. However, in many settings, the analyst's data is a sample of a continuous process, and their interests lie in extrapolating from current data to make broader statements about a model.
Note that when we are given the sample, we can always make the coarsening $\bar{\mathcal{F}}_{\sigma, \varepsilon}$ sufficiently fine to ensure that each observation $x_i$ is contained in its own partition of the coarsening. However, if samples are drawn continuously over an infinite domain, there is no such level of partitioning that guarantees such a result in general. As a result, when only the sample size is known, we show two results. First, the probability of a discrepancy between the Rademacher complexity and a sample analog of the numerator of restrictiveness can be made arbitrarily small.

\begin{lemma}\label{lemma:probably_the_same}
    Assume $d_n$ is the correlation discrepancy in Lemma \ref{lemma:inf_sup_correlation}. Given an atomless distribution $P_X$, a sample size $n$, and a tolerance level $\delta>0$, there exists an $\varepsilon_{n, \delta}>0$ such that
    $$\text{Pr}_{S \sim P_X^n}\left\{ \tilde{R}_n(\mathcal{F}_\Theta, S) \neq 1 - \tilde{e}_n\left(\mathcal{F}_\Theta, \bar{\mathcal{F}}_{\sigma, \varepsilon_{n, \delta}}, d_n, S \right) \right\} < \delta.$$
\end{lemma}

\begin{proof}
    Based on the previous results, setting a sufficiently low $\varepsilon$ will guarantee equality between Rademacher complexity and (one minus) the restrictiveness numerator. Hence, it suffices to find an $\varepsilon$ small enough to guarantee the equality at least $(1-\delta)$ of the time.

    Again, since the discrepancy (pun intended) stems primarily from situations where $x_i$ and $x_j$ are assigned the same bin, it becomes sufficient to examine the probability that such a collision occurs. Let $p_{max}=\max_{p \in P_{X_\varepsilon}}p$. Applying Boole's inequality,

    \begin{align*}
      \Pr_{S\sim P_X^n}(\text{collision})
      &\leq\sum_{i<j}\Pr(X_{i,\varepsilon}=X_{j,\varepsilon})\\
      &=\binom{n}{2}\sum_{A\in\mathcal X_\varepsilon}P_X(A)^2\\
      &\leq\binom{n}{2}p_{max}.
    \end{align*}
    Here $X_{i,\varepsilon}$ is the cell containing $X_i$. Note that $n$ is fixed. Moreover,
    \[
    p_{max}^2\leq\sum_{A\in\mathcal X_\varepsilon}P_X(A)^2
    \leq\Pr(\|X_1-X_2\|_\infty\leq\varepsilon)
    \longrightarrow\Pr(X_1=X_2)=0
    \]
    as $\varepsilon\downarrow0$, since $P_X$ is atomless. Thus a sufficiently small $\varepsilon$ makes the collision probability less than $\delta$. For $n=1$, no collision is possible and every partition gives the equality.
\end{proof}

Second, the value of this difference can also be made arbitrarily small.

\begin{proposition}\label{prop:bounded_difference}
    Assume $d_n$ is the correlation discrepancy in Lemma \ref{lemma:inf_sup_correlation}. Given an atomless distribution $P_X$, a sample size $n$, and a tolerance level $\delta > 0$, there exists an $\varepsilon_{n, \delta}>0$ such that $$\left| R_n(\mathcal{F}_\Theta) - \left[ 1 - e_n\left(\mathcal{F}_\Theta, \bar{\mathcal{F}}_{\sigma, \varepsilon_{n, \delta}}, d_n\right) \right] \right| < \delta.$$
\end{proposition}

\begin{proof}
    For each coarsening $\varepsilon$, define $\Delta(S) = \tilde{R}_n(\mathcal{F}_\Theta, S) - \left[1 - \tilde{e}_n(\mathcal{F}_\Theta, \bar{\mathcal{F}}_{\sigma, \varepsilon}, d_n, S)\right]$. Because $\mathcal{Y}$ is bounded and because $|g(x, \theta; f)| \leq G(x)$,
    $$\left|\Delta(S)\right| \leq 1 + \frac{1}{n}\sum_{i=1}^nG(x_i).$$

    Denote the right-hand side as $M(S)$. Define a sequence of $\varepsilon_1 > \varepsilon_2 > \dots 0$ such that $\lim_{\eta\rightarrow\infty}\varepsilon_\eta = 0$. Thus, for a given $\eta$,

    \begin{align*}
        \left| R_n(\mathcal{F}_\Theta) - \left[ 1 - e_n\left(\mathcal{F}_\Theta, \bar{\mathcal{F}}_{\sigma, \varepsilon_\eta}, d_n\right) \right] \right| & = \left| E_S\left[ \Delta(S) \right] \right| \\
        & \leq E_S\left[\left| \Delta(S) \right|\right] \\
        & \leq E_S\left[ I(\eta\text{-collision}) \cdot M(S) \right]\\
    \end{align*}

    Since $P_X$ is atomless, a fixed finite sample has distinct observations almost surely, and sufficiently small $\varepsilon_\eta$ separates them. Thus $I(\eta\text{-collision})\to0$ almost surely. By assumption on $G$, $E_S[M(S)]<\infty$, so dominated convergence makes the last expression tend to zero. A sufficiently large $\eta$ therefore gives the desired bound.
\end{proof}

Finally, in Lemma \ref{lemma:convergence} we show that the average approximation error is truly an empirical analog of approximation error. The proof of Lemma \ref{lemma:convergence} relies on two auxiliary lemmas, which we state in Appendix \ref{A.1}. Recall also our assumptions discussed in Section \ref{sec:notation}, namely that $g$ is bounded by an envelope function $G$ with finite mean, $g$ is continuous in $\theta$ with probability one, and $\Theta$ is compact.

\begin{lemma}\label{lemma:convergence}
    Fix a parametric model family $\mathcal{F}_\Theta$ and let $\mathcal{F}$ be a nonempty finite set of functions. Let $S=(x_1, \dots, x_n)$ be a random sample drawn iid from $P_X$, and let

    $$d_n(f_\theta, f) = \frac{1}{n}\sum_{i=1}^ng(x_i, \theta; f),$$

    Under the assumptions of Section \ref{sec:notation},
    \[
    E_{\lambda_{\mathcal F}}E_S
       \left[\inf_{\theta\in\Theta}d_n(f_\theta,f)\right]
    \longrightarrow
    E_{\lambda_{\mathcal F}}
       \left[\inf_{\theta\in\Theta}d(f_\theta,f)\right].
    \]

\end{lemma}

\begin{proof}
    Lemmas \ref{lemma:uniform} and \ref{lemma:m_estimator} give, for each $f\in\mathcal F$,
    \[
    E_S\left|\inf_{\theta\in\Theta}d_n(f_\theta,f)
              -\inf_{\theta\in\Theta}d(f_\theta,f)\right|
    \leq E_S\left[\sup_{\theta\in\Theta}|d_n(f_\theta,f)-d(f_\theta,f)|\right]
    \longrightarrow0.
    \]
    Since $\mathcal F$ is finite, averaging over $\lambda_{\mathcal F}$ and applying the triangle inequality proves the result.
\end{proof}

By itself, Lemma \ref{lemma:convergence} shows that restrictiveness is widely applicable in empirical work, since its estimation is robust to not knowing $P_X$ under reasonable assumptions. This complements the analysis in \cite{fudenberg2023restrictiveness} that samples from $\mathcal{F}$ instead of knowing $\lambda_{\mathcal{F}}$. Overall, our results show that for Rademacher complexity, there exists an analogous function class that fits within an empirical analogy of restrictiveness. Such a function class then can be tightly approximated by finite coarsenings whose average approximation error converges to approximation error in the limit. 

The argument for Lemma \ref{lemma:convergence} also provides a microfoundation for Proposition 2 of \cite{fudenberg2026model}, which interprets restrictiveness as the normalized limit of a noise-free average-case learning curve. In our setting, minimizing sample discrepancy yields population discrepancy converging in expectation to the smallest value attainable by the model, providing sufficient conditions for their risk-consistency assumption.

\subsection{Offset Rademacher complexity}

In this section, we show that similar results can be obtained with respect to \textit{offset} Rademacher complexity. We make two quick observations. First, the results show that
restrictiveness connects to more measures than solely Rademacher complexity. Furthermore, it provides alternative environments where the interpretation of restrictiveness can benefit from results on another measure (e.g., \citealp{liang2015learning}); in this case, offset processes enter bounds for square-loss learning under that paper's learning and geometric conditions.

Propositions \ref{prop:offset_equivalence_in_sample} and \ref{prop:offset_bounded_difference} provide analogous results regarding connections between (empirical) offset Rademacher complexity and (empirical) approximation error. The proofs are analogously similar, and are provided in Appendices \ref{A.2} and \ref{A.3} for completeness. The connection relies on two generalizations. First, because offset Rademacher complexity contains a free parameter $c\geq0$, the exact affine transformation between offset Rademacher complexity and approximation error depends on the parameter. Second, this parameter $c$ also informs the analogous eligible set of functions.

\begin{proposition}\label{prop:offset_equivalence_in_sample}
    Assume a discrepancy function $d_n(f, g; c) = \frac{c}{n}\sum_{i=1}^n\left[ f(x_i) - g(x_i) \right]^2.$ Given a sample of unique covariates $S = (x_1, \dots, x_n)$, there exists a function class $\frac{1}{2c}\bar{\mathcal{F}}_{\sigma, \varepsilon} := \left\{\frac{f}{2c}:f\in\bar{\mathcal F}_{\sigma,\varepsilon}\right\}$ such that \[ \tilde{R}_n^{\text{off}}(\mathcal{F}_\Theta, S, c) = \frac{1}{4c} - \tilde{e}_n\left(\mathcal{F}_\Theta, \frac{1}{2c}\bar{\mathcal{F}}_{\sigma, \varepsilon}, d_n, S\right) \]
    for all {$c\geq1/2$}.
\end{proposition}

\begin{proposition}\label{prop:offset_bounded_difference}
    Fix {$c\geq1/2$} and the discrepancy in Proposition \ref{prop:offset_equivalence_in_sample}. Given an atomless distribution $P_X$, a sample size $n$, and a tolerance level $\delta>0$, there exists an $\varepsilon_{n, \delta}>0$ such that
    $$\left| R_n^{\text{off}}(\mathcal{F}_\Theta, c) - \left[\frac{1}{4c} - e_n\left(\mathcal{F}_\Theta, \frac{1}{2c}\bar{\mathcal{F}}_{\sigma, \varepsilon_{n, \delta}}, d_n\right)\right] \right| < \delta.$$
\end{proposition}

The restriction $c\geq1/2$ ensures that the eligible outcomes $\pm1/(2c)$ belong to our common outcome space $[-1,1]$. For any fixed $c>0$, both propositions continue to hold with model predictions still in $[-1,1]$ if the outcome space is enlarged to $[-M_c,M_c]$, where $M_c=\max\{1,1/(2c)\}$. When $c=0$, offset complexity equals ordinary Rademacher complexity, and Propositions \ref{prop:equivalence_in_sample} and \ref{prop:bounded_difference} give the corresponding connection using correlation discrepancy.

\section{Restrictiveness and Generalization Error}\label{sec:generalization_error}

Rademacher complexity bounds generalization error, but its limiting value need not distinguish models that permit different behavior. For example, binary-choice models with finite VC dimension all have limiting Rademacher complexity zero, even when one model strictly extends another.\footnote{{Offset Rademacher complexity, connected to squared-error discrepancy by Propositions \ref{prop:offset_equivalence_in_sample} and \ref{prop:offset_bounded_difference}, can retain this limitation. For fixed $c\geq1/2$, if the model contains the zero function, then $0\leq R_n^{\mathrm{off}}(\mathcal F_\Theta,c)\leq R_n(\mathcal F_\Theta)$. For $\{-1,1\}$-valued models, $R_n^{\mathrm{off}}(\mathcal F_\Theta,c)=R_n(\mathcal F_\Theta)-c$. Thus models with vanishing Rademacher complexity share an offset limit of zero in the first case and $-c$ in the second.}} This limitation is economically relevant: in the finite-state setting with fixed utility studied by \cite{basu2020falsifiability}, both EUT and CEU have finite VC dimension and hence vanishing Rademacher complexity. The limiting measure therefore does not distinguish EUT from CEU, even though CEU accommodates the ambiguity-averse choices in the Ellsberg paradox that EUT rules out.

This limitation extends to other economically relevant environments as well. For example, in a two-good demand setting \citep{beigman2006learning}, Cobb--Douglas and the larger class of constant elasticity of substitution (CES) functions for relative demand both have limiting Rademacher complexity zero when substitution elasticities are positive and uniformly bounded above. Similarly, with uniformly bounded predictions, affine models and rectified linear unit (ReLU) neural networks of fixed architecture both have limiting Rademacher complexity zero, although the networks permit nonlinear, piecewise-affine responses \citep{farrell2021deep}.\footnote{These approaches extend VC dimension methods to real-valued classes using pseudodimension and covering numbers. Finite pseudodimension bounds sample covering numbers and, through Massart's lemma \citep{mitzenmacher2017probability}, gives upper bounds on Rademacher complexity that converge to zero as the number of observations grows.}

Choosing economically meaningful eligible functions, e.g. all CEU functions, addresses that limitation but also raises the question of whether a cardinal interpretation of restrictiveness remains. However, we show that restrictiveness provides useful cardinal meaning in economically relevant environments. The following result provides one set of sufficient conditions for such an interpretation.

We start by making additional assumptions on our environment. First, we assume that our input space $X = (Z, T) \in \mathcal{X}$, where $Z \in \mathbb{R}^{m-1}$ and $T \in \mathbb{R}$. We assume that, conditional on $Z$, all $f_\theta \in \mathcal{F}_\Theta$ and almost surely all eligible functions $f \sim \lambda_\mathcal{F}$ are {weakly nondecreasing in $T$}. This condition is commonly satisfied in economically relevant environments.\footnote{For example, consider the application of certainty equivalents from \cite{fudenberg2023restrictiveness} and \cite{fudenberg2026model}. Subjects are shown a binary lottery with payments $(\underline{z}, \bar{z})$ with probabilities $(1-p, p)$. The considered set of eligible functions satisfies monotonicity with respect to first order stochastic dominance, and the models considered are expected utility theory and non-expected utility theories, all of which also satisfy the same monotonicity condition. {For this result, take $T=p$ and $Z=(\underline z,\bar z)$, with $\underline z<\bar z$, and normalize a certainty equivalent $a$ to $2(a-\underline z)/(\bar z-\underline z)-1$. With continuous strictly increasing utility and increasing probability weighting, nonlinear EUT and cumulative prospect theory satisfy the required ordering in $p$.}} We further assume that $g(x, \theta; f) = \ell\left[ f_\theta(x), f(x) \right]$ for some {Borel measurable }pointwise $\ell: \mathcal{Y} \times \mathcal{Y} \rightarrow \mathbb{R}_+$. We assume the loss is convex in its first argument, and $\Delta(v) = \ell(1, v) - \ell(-1, v)$ is nonincreasing in $v$. Examples include squared error, mean absolute error, and Huber loss. Pinball loss is an asymmetric example {(\citealp[Chapter 1]{koenker2005quantile})}. We also strengthen the envelope assumption from Section \ref{sec:notation_and_objects} and assume that \[{\max\{|g(x,\theta;f)|,\ell(-1,f(x)),\ell(1,f(x))\}\le G(x),\qquad E_{P_X}[G(X)]<\infty}.\] {For almost every $f$ under $\lambda_{\mathcal F}$, the bound holds for all $\theta$ with probability one under $P_X$.}

Here, we define the \textit{expected uniform generalization error in discrepancy}  as \[{\operatorname{Gen}_n(\mathcal F_\Theta,\mathcal F,d_n)=\mathbb E_{\lambda_{\mathcal F}}\mathbb E_S\left[\sup_{f_\theta\in\mathcal F_\Theta}\{d(f_\theta,f)-d_n(f_\theta,f)\}\right]}.\] In words, this expression measures, across samples and eligible functions, the worst-case difference between a model's discrepancy to the population and its discrepancy to a given sample.\footnote{Note that for a fixed eligible function $f$, this is the sample error functional of \cite{mitzenmacher2017probability} applied to the set of possible discrepancies {$\{g(\cdot,\theta;f):\theta\in\Theta\}$}. This is a standard unit of analysis (see e.g., \citealp{bartlett2002rademacher}).}

{To relate this expected gap to average approximation error, we construct a common upper bound on population discrepancy.}
Let $Q_z$ denote the conditional distribution of $T$ given $Z=z$, and define \[ d_z(h,f)=\int\ell(h(z,t),f(z,t))\,Q_z(dt) \] for any measurable $[-1,1]$-valued prediction rule $h$. For each $f,z$, we compare the rules that predict a constant one or minus-one at all values of $T$. Their larger conditional population discrepancy determines

\begin{equation}\label{eq:endpt-benchmark}
    \begin{gathered}
        B_f(z)=\max\{d_z(-1,f),d_z(1,f)\},\qquad B(f)=\mathbb E_ZB_f(Z),\\
        D=\mathbb E_{f,Z}B_f(Z),\qquad r_n^{\mathrm{end}}=\frac{e_n}{D}.
    \end{gathered}
\end{equation}

The envelope gives $D\le\mathbb E G(X)<\infty$; like \cite{fudenberg2023restrictiveness}, we assume $D>0$ to define the normalized index.

\begin{proposition}\label{prop:endpt-bound}
    Under the stated conditions,
    \begin{equation}\label{eq:endpt-bound}
        0\le\operatorname{Gen}_n(\mathcal F_\Theta,\mathcal F,d_n)
        \le D-e_n(\mathcal F_\Theta,\mathcal F,d_n)
        =D\{1-r_n^{\mathrm{end}}(\mathcal F_\Theta,\mathcal F,d_n)\}.
    \end{equation}
    Moreover, $0\le e_n\le e\le D$, so $r_n^{\mathrm{end}}\in[0,1]$.
\end{proposition}

The proof is in Appendix~{\ref{A.4}}. The result shows that restrictiveness produces a similar type of bound to Rademacher complexity in the appropriate environment. {This cardinal interpretation relies on precisely the choices that generalize Rademacher complexity to other measures within the restrictiveness family: the discrepancy function and the eligible set.} {Example~\ref{ex:cara-crra-short} illustrates this interpretation for a misspecified model of risk preferences.}

\begin{example}[Misspecification across stakes]\label{ex:cara-crra-short}
Consider the choice between two lotteries - one that pays $\$c$ for sure, and another that pays $\$3c$ with 50\% probability. Let the range be $\{-1, 1\}$ with $1$ corresponding to the sure lottery and $-1$ corresponding to the 50-50 gamble.
The value of $c$ is distributed according to $C\sim U[0.5, 1.5]$, which we take as our input $\mathcal{X}$. We consider a singleton eligible set defined by expected utility with constant absolute risk aversion (CARA) $u(x)=-e^{-ax}$, with $a=\log((1+\sqrt5)/2)$. This rule switches from the lottery to the sure payment at $c=1$.

We consider a parametrically misspecified model of expected utility with constant relative risk aversion (CRRA) $u(x)=\frac{x^{1-\rho}}{1-\rho}$, considering $\rho\in[0,1)$. Note in these problems, CRRA predicts constant behavior across $c$'s dependent on its parameter value: if $\rho \geq 1-\frac{\log(2)}{\log(3)}$, then the sure payoff is chosen (breaking ties in its favor); otherwise, the gamble is chosen. The two distinct prediction rules admit a finite discrete parameterization, satisfying the compactness and continuity assumptions. If we consider disagreement loss $\ell(y,v)=|y-v|/2$, both constant rules have population discrepancy $D=e=1/2$. If
$K\sim\mathrm{Bin}(n,1/2)$ sampled choices from our eligible function have $c\geq1$, then $\widetilde e_n=\min\{K/n,1-K/n\}$, giving the exact endpoint bound:
\[
 \operatorname{Gen}_n=D-e_n=e-e_n
 =\frac12(1-r_n^{\mathrm{end}})
 =\mathbb E\left|\frac Kn-\frac12\right|.
\]
In this simple example, the bound is satisfied with equality. Here generalization error is exactly the expected understatement of population misspecification from fitting the sample. Generalization error is $3.98$ percentage points of disagreement at $n=100$ and vanishes as $n$ grows.
\end{example}

We note that this connection can also apply to completeness \citep{fudenberg2022measuring}. This measure evaluates predictive performance in a given population. Viewing each eligible function as a noiseless population, the gap between empirical and population completeness is a normalized discrepancy gap. A simple corollary of Proposition~\ref{prop:endpt-bound} bounds the expected overstatement of completeness, weighted by baseline discrepancy. Suppose $d(f_{base},f)>0$ for $\lambda_{\mathcal F}$-almost every $f$, and define population and empirical completeness using this same population baseline:
\[
\begin{aligned}
\kappa(f)&=\frac{d(f_{base},f)-\inf_{\theta\in\Theta}d(f_\theta,f)}{d(f_{base},f)-d(f,f)}
=1-\frac{\inf_{\theta\in\Theta}d(f_\theta,f)}{d(f_{base},f)},\\[0.6em]
\widetilde\kappa_n(f,S)&=1-\frac{\inf_{\theta\in\Theta}d_n(f_\theta,f)}{d(f_{base},f)}.
\end{aligned}
\]
Then
\[
\begin{aligned}
0&\le \mathbb E_{\lambda_{\mathcal F}}\!\left[d(f_{base},f)\,\mathbb E_S\left[\widetilde\kappa_n(f,S)-\kappa(f)\right]\right]
=e-e_n\\
&\le \operatorname{Gen}_n\le D\{1-r_n^{\mathrm{end}}\}.
\end{aligned}
\]
The intermediate inequality follows because the difference of infima is bounded by the supremum of the difference.

\section{Conclusion}\label{sec:interpretation}

This paper connects restrictiveness to Rademacher complexity and its offset. For particular choices of the eligible set and discrepancy function, these measures are affine transformations of empirical approximation error, the sample counterpart of the restrictiveness numerator. We establish an exact connection when the sample is known and an arbitrarily close approximation in expectation when only its size is known. We also give conditions under which average approximation error consistently estimates the population numerator, extending the use of restrictiveness to settings where the covariate distribution is unknown. Using this connection, we can also easily augment the axioms of restrictiveness to provide an axiomatization of Rademacher complexity.

The connection highlights the economic importance of the choices available within the restrictiveness family. There exist economically relevant environments in which limiting Rademacher complexity does not distinguish models that permit different behavior. With an economically meaningful eligible set, however, restrictiveness can avoid this unidentifiability while maintaining a cardinal interpretation of bounding discrepancy-dependent generalization error.


Looking forward, our work shows a connection between one measure within economic theory and statistical learning theory, but there are likely further connections. Furthermore, we show only sufficient conditions for our generalization error bound; extending the result beyond the monotonicity conditions used here would allow this interpretation to apply to a wider range of economic environments. These directions would prove interesting for future work.

\newpage
\appendix
\setcounter{page}{1}
\renewcommand\thefigure{\thesection.\arabic{figure}}
\renewcommand\thetable{\thesection.\arabic{figure}}
\renewcommand\thelemma{\thesection.\arabic{lemma}}
\renewcommand\thecorollary{\thesection.\arabic{corollary}}

\section{Appendix: Supporting Results}

\subsection{Auxiliary Lemmas for Lemma \ref{lemma:convergence}}\label{A.1}

We use the usual completions of the probability measures in all expectation statements; this includes subsets of probability-zero events without changing their probabilities.

\begin{lemma}\label{lemma:uniform}
    Fix a parametric model family $\mathcal{F}_\Theta$ and fix $f \in \mathcal{F}$. Let $S=(x_1, \dots, x_n)$ be a random sample drawn iid from $P_X$. Let the empirical discrepancy be $d_n(f_\theta, f) = \frac{1}{n}\sum_{i=1}^ng(x_i, \theta; f)$.
    Assume $g(x,\theta;f)$ is jointly measurable and that:
    \begin{enumerate}
        \item There exists an envelope function $G(x)$ such that $|g(x, \theta; f)| \leq G(x)$ for all $\theta \in \Theta$ and $P_X$-almost every $x$ and $E[G(X)]<\infty$.
        \item $g(X, \theta; f)$ is continuous in $\theta$ with probability one, and $\Theta$ is a nonempty compact metric space.
    \end{enumerate}

    Then
    \[
    E_S\left[\sup_{\theta\in\Theta}
             |d_n(f_\theta,f)-d(f_\theta,f)|\right]\longrightarrow0.
    \]
\end{lemma}

\begin{proof}
    {Choose a countable dense subset $\Theta_0\subseteq\Theta$. Parameter continuity makes $d_n(f_\theta,f)$ continuous in $\theta$ with probability one, while the envelope and dominated convergence give continuity of $d(f_\theta,f)$. Their infima and the suprema of $d-d_n$ and $|d-d_n|$ can therefore be taken over $\Theta_0$, giving measurable versions of these quantities.}

    Let $U_n=\sup_{\theta\in\Theta}|d_n(f_\theta,f)-d(f_\theta,f)|$. Write $\bar G_n=n^{-1}\sum_{i=1}^nG(x_i)$ and $\mu=E_{P_X}[G(X)]$. By Theorem 18.2 of \cite{hansen2022probability}, $U_n\xrightarrow{P}0$. The envelope gives $0\leq U_n\leq\bar G_n+\mu$. Thus, for any $\eta>0$,
    \begin{align*}
    E_S[U_n]
    &\leq\eta+E_S\left[(\bar G_n+\mu)I(U_n>\eta)\right]\\
    &\leq\eta+2\mu\Pr_{S\sim P_X^n}(U_n>\eta)
             +E_S\left[|\bar G_n-\mu|\right].
    \end{align*}
    The probability term converges to zero, and the last term converges to zero by the $L^1$ law of large numbers for the integrable envelope. Hence $\limsup_{n\to\infty}E_S[U_n]\leq\eta$. Since $\eta>0$ is arbitrary, the result follows.
\end{proof}

This encompasses multiple types of loss functions. For example, bounded outcomes give constant envelopes for squared and absolute error.

\begin{lemma}\label{lemma:m_estimator}
    Fix a parametric model family $\mathcal{F}_\Theta$ and fix $f \in \mathcal{F}$. Let $S=(x_1, \dots, x_n)$ be a random sample drawn iid from $P_X$.
    Whenever the infima below are finite,
    \[
    \left|\inf_{\theta\in\Theta}d_n(f_\theta,f)
              -\inf_{\theta\in\Theta}d(f_\theta,f)\right|
    \leq\sup_{\theta\in\Theta}|d_n(f_\theta,f)-d(f_\theta,f)|.
    \]
\end{lemma}

\begin{proof}
    Write $U=\sup_{\theta\in\Theta}|d_n(f_\theta,f)-d(f_\theta,f)|$. If $U$ is finite, then
    $d(f_\theta,f)-U\leq d_n(f_\theta,f)\leq d(f_\theta,f)+U$ for every $\theta$. Taking infima proves both inequalities. If $U$ is infinite, the result is immediate.
\end{proof}

\subsection{Proof of Proposition \ref{prop:offset_equivalence_in_sample}}\label{A.2}

First, note that

\begin{align*}
    \frac{\sum_{i=1}^n[\sigma_if(x_i)-cf(x_i)^2]}{n} & = \frac{1}{4c}-\frac{c}{n}\sum_{i=1}^n\left[ f(x_i) - \frac{\sigma_i}{2c} \right]^2
\end{align*}

This can be shown by completing the square for each term noting that $\sigma_i^2=1$. Thus, empirical offset Rademacher complexity can be shown to equal

\begin{align*}
    \tilde{R}^{\text{off}}_n(\mathcal{F}_{\Theta},S,c) & = \mathbb{E}_\sigma \left[ \underset{f_\theta \in \mathcal{F}_{\Theta}}{\text{sup}}\ \frac{1}{n} \sum^n_{i=1} [\sigma_i \cdot f_\theta(x_i) - c f_\theta(x_i)^2] \right] \\
    & = \mathbb{E}_\sigma\left\{ \underset{f_\theta \in \mathcal{F}_{\Theta}}{\text{sup}} \frac{1}{4c} - \frac{c}{n}\sum_{i=1}^n\left[ f_\theta(x_i) - \frac{1}{2c}\sigma_i \right]^2\right\}
\end{align*}

\begin{lemma}\label{lemma:offset_sigma_function}
    Given a sample of unique covariates $S = (x_1, \dots, x_n)$, there exists a function class $\bar{\mathcal{F}}_{\sigma, \varepsilon}$ such that
    \begin{align*}
    & \mathbb{E}_\sigma \left[ \underset{f_\theta \in \mathcal{F}_{\Theta}}{\text{sup}}\ \frac{1}{n} \sum^n_{i=1} [\sigma_i \cdot f_\theta(x_i) - c f_\theta(x_i)^2] \right] \\
    & = \mathbb{E}_{f \sim \lambda_{\bar{\mathcal{F}}_{\sigma, \varepsilon}}} \left[ \underset{f_\theta \in \mathcal{F}_{\Theta}}{\text{sup}}\ \frac{1}{n} \sum^n_{i=1} [f(x_i) \cdot f_\theta(x_i) - c f_\theta(x_i)^2] \right]
    \end{align*}
    for all $c>0$.
\end{lemma}

\begin{proof}
    The proof is the same as Lemma \ref{lemma:sigma_function}; choosing $\varepsilon < \min_{i\neq j}||x_i-x_j||_\infty$ for $n\geq2$ (and any partition for $n=1$) guarantees that the interval coarsening will be fine enough to place each $x_i$ into a unique partition.
\end{proof}

For its convenience later, we can instead write a version of Lemma \ref{lemma:offset_sigma_function} that takes the uniform distribution over $\frac{1}{2c}\bar{\mathcal{F}}_{\sigma, \varepsilon}:= \left\{\frac{f}{2c}:f\in\bar{\mathcal F}_{\sigma,\varepsilon}\right\}$.

\begin{corollary}\label{corollary:offset_sigma_function}
    Given a sample of unique covariates $S = (x_1, \dots, x_n)$, there exists a function class $\bar{\mathcal{F}}_{\sigma, \varepsilon}$ such that
    \begin{align*}
    & \mathbb{E}_\sigma \left[ \underset{f_\theta \in \mathcal{F}_{\Theta}}{\text{sup}}\ \frac{1}{n} \sum^n_{i=1} [\sigma_i \cdot f_\theta(x_i) - c f_\theta(x_i)^2] \right] \\
    & = \mathbb{E}_{f \sim \lambda_{\frac{1}{2c}\bar{\mathcal{F}}_{\sigma, \varepsilon}}} \left[ \underset{f_\theta \in \mathcal{F}_{\Theta}}{\text{sup}}\ \frac{1}{4c} - \frac{c}{n}\sum_{i=1}^n\left[ f_\theta(x_i) - f(x_i) \right]^2 \right]
    \end{align*}
    for all $c\geq1/2$.
\end{corollary}

\begin{lemma}\label{lemma:inf_sup_squared_error}
    Fix a sample $S=(x_1, \dots, x_n)$, a mapping $g \in \frac{1}{2c}\bar{\mathcal{F}}_{\sigma, \varepsilon}$, and assume a discrepancy function of the form $d_n: \bar{\mathcal{F}} \times \frac{1}{2c}\bar{\mathcal{F}}_{\sigma, \varepsilon} \rightarrow \mathbb{R}_+$ such that
    $$d_n(f, g; c) = \frac{c}{n}\sum_{i=1}^n\left[ f(x_i) - g(x_i) \right]^2.$$
    Then, $$ \sup_{f_\theta \in \mathcal{F}_\Theta} \frac{1}{4c} - \frac{c}{n}\sum_{i=1}^n\left[ f_\theta(x_i) - g(x_i)\right]^2 = \frac{1}{4c} - \inf_{f_\theta \in \mathcal{F}_\Theta} d_n(f_\theta, g; c)$$
    for all $c\geq1/2$.
\end{lemma} 

Finally, the proof of Proposition \ref{prop:offset_equivalence_in_sample} follows as a consequence of the two lemmas. Again setting $\varepsilon < \min_{i\neq j}||x_i-x_j||_\infty$ when $n\geq2$ (and using any partition when $n=1$):

\begin{align*}
    \tilde{R}^{\text{off}}_n(\mathcal{F}_{\Theta},S,c) & = \mathbb{E}_\sigma \left[ \underset{f_\theta \in \mathcal{F}_{\Theta}}{\text{sup}}\ \frac{1}{n} \sum^n_{i=1} [\sigma_i \cdot f_\theta(x_i) - c f_\theta(x_i)^2] \right] \\
    & = \mathbb{E}_{f \sim \lambda_{\frac{1}{2c}\bar{\mathcal{F}}_{\sigma, \varepsilon}}}\left\{ \underset{f_\theta \in \mathcal{F}_{\Theta}}{\text{sup}} \frac{1}{4c} - \frac{c}{n}\sum_{i=1}^n\left[ f_\theta(x_i) - f(x_i)\right]^2\right\} \\
    & = \frac{1}{4c} - \mathbb{E}_{f \sim \lambda_{\frac{1}{2c}\bar{\mathcal{F}}_{\sigma, \varepsilon}}}\left[\inf_{f_\theta \in \mathcal{F}_\Theta} d_n(f_\theta, f; c)\right] \\
    & = \frac{1}{4c} - \tilde{e}_n\left(\mathcal{F}_\Theta, \frac{1}{2c}\bar{\mathcal{F}}_{\sigma, \varepsilon}, d_n, S\right)
\end{align*}

\subsection{Proof of Proposition \ref{prop:offset_bounded_difference}}\label{A.3}

\begin{lemma}\label{lemma:offset_probably_the_same}
    Fix $c\geq1/2$ and the discrepancy in Proposition \ref{prop:offset_equivalence_in_sample}. Given an atomless distribution $P_X$, a sample size $n$, and a tolerance level $\delta>0$, there exists an $\varepsilon_{n, \delta}>0$ such that
    $$\text{Pr}_{S \sim P_X^n}\left\{ \tilde{R}_n^{\text{off}}(\mathcal{F}_\Theta, S, c) \neq \frac{1}{4c} - \tilde{e}_n\left(\mathcal{F}_\Theta, \frac{1}{2c}\bar{\mathcal{F}}_{\sigma, \varepsilon_{n,\delta}}, d_n, S\right) \right\} < \delta.$$
\end{lemma}

\begin{proof}
    The proof is identical to Lemma \ref{lemma:probably_the_same}.
\end{proof}

The proof follows the argument of Proposition \ref{prop:bounded_difference}. For each coarsening $\varepsilon$, define
\[
\Delta(S)=\tilde R_n^{\mathrm{off}}(\mathcal F_\Theta,S,c)
 -\left[\frac{1}{4c}-\tilde e_n\!\left(\mathcal F_\Theta,
       \frac{1}{2c}\bar{\mathcal F}_{\sigma,\varepsilon},d_n,S\right)\right].
\]
Both expressions in $\Delta(S)$ are bounded above by $1/(4c)$ and below by $-c$. The upper bound follows by completing the square. For the lower bound, fix any $f_\theta\in\mathcal F_\Theta$ before taking the expectation over the signs. Under either distribution, the linear term has mean zero and $-c n^{-1}\sum_{i=1}^n f_\theta(x_i)^2\geq-c$. Hence $|\Delta(S)|\leq c+1/(4c)$.

Moreover, $\Delta(S)=0$ whenever observations occupy distinct partition cells. Therefore,
\[
\left|E_S[\Delta(S)]\right|
\leq E_S[|\Delta(S)|]
\leq\left(c+\frac{1}{4c}\right)
       \Pr_{S\sim P_X^n}(\text{collision}).
\]
The last expression tends to zero as $\varepsilon\downarrow0$ by Lemma \ref{lemma:probably_the_same}, proving the result.

\subsection{Proof of Proposition \ref{prop:endpt-bound}}\label{A.4}

The envelope inherited from Section~\ref{sec:notation}, extended to the
endpoint losses, gives
$B_f(z)\le\mathbb E[G(X)\mid Z=z]$ and hence $D\le\mathbb E G(X)<\infty$.
Joint measurability of predictions and losses also permits integration over eligible functions and conditional covariate distributions.

Fix an eligible $f$, a parameter $\theta$, and a conditioning value $z$ for
which the assumptions hold. Write $a(t)=f_\theta(z,t)$, $v(t)=f(z,t)$,
and $Q_z=P_{T\mid Z=z}$. Convexity in the prediction gives
\[
 \ell(a(t),v(t))
 \le\ell(-1,v(t))+\frac{a(t)+1}{2}\Delta(v(t)).
\]
The function $a$ is nondecreasing and $\Delta\circ v$ is nonincreasing.
For independent $T,T'\sim Q_z$,
\begin{align*}
 2\operatorname{Cov}_{Q_z}(a,\Delta\circ v)
 =\mathbb E\big[&(a(T)-a(T'))\\
 &\quad\cdot(\Delta(v(T))-\Delta(v(T')))\big]\le0.
\end{align*}
It follows that
\begin{align*}
 d_z(f_\theta,f)
 &\le d_z(-1,f)
   +\frac{\mathbb E_{Q_z}a+1}{2}\{d_z(1,f)-d_z(-1,f)\}\\
 &\le\max\{d_z(-1,f),d_z(1,f)\}=B_f(z).
\end{align*}
Integrating over $Z$ yields $d(f_\theta,f)\le B(f)$.
The conditional distribution of $T$ can vary with $z$; independence of
$Z$ and $T$ has not been used.

For almost every eligible function and sample,
\[
 \sup_{\theta\in\Theta}\{d(f_\theta,f)-d_n(f_\theta,f)\}
 \le B(f)-\inf_{\theta\in\Theta}d_n(f_\theta,f).
\]
{The countable set $\Theta_0$ in the proof of Lemma \ref{lemma:uniform} can be chosen independently of $f$. By joint measurability in Section~\ref{sec:notation}, taking these extrema over $\Theta_0$ gives jointly measurable versions in $(f,S)$. They agree with the original extrema with probability one under $P_X^n$ for almost every $f$ under $\lambda_{\mathcal F}$. The common integrable envelope makes the resulting expectations finite.}
Taking expectations gives $\operatorname{Gen}_n\le D-e_n$.
The expected supremum is nonnegative because it dominates the gap of any
fixed model element, whose expectation is zero. Nonnegative losses imply
$e_n\ge0$, while moving an infimum outside expectation gives $e_n\le e$.
Finally, the population comparison gives $e\le D$.
This proves the proposition. The optimization remains over the original
parameter space and is never performed separately within conditioning
values.\hfill$\square$

\section{Axiomatization of Rademacher Complexity}\label{sec:axiomatization}

Proposition \ref{prop:equivalence_in_sample} shows that for a particular eligible set and a particular discrepancy function, empirical Rademacher complexity equals one minus empirical approximation error. {This appendix gives an axiomatic characterization:} the eligible set, the discrepancy function, and the aggregation across synthetic outcomes are each pinned down by interpretable conditions, so that \textit{any} empirical approximation error satisfying them must be an affine transformation of empirical Rademacher complexity (Proposition \ref{prop:axiomatization}). The axioms imposed on the aggregation are finite-sample counterparts of the axioms with which \cite{fudenberg2023restrictiveness} characterize approximation error; the axioms imposed on the eligible set and on the discrepancy function are new, and they are what select the Rademacher structure --- all binary synthetic outcome vectors, evaluated by one minus correlation and aggregated uniformly --- within the restrictiveness framework.

\subsection{Environment}\label{sec:axiom_environment}

We follow the notation of the main text, with the following conventions and deviations, each of which is necessary for the exercise. Throughout this appendix the sample $S = (x_1, \dots, x_n)$ is fixed, $n \geq 1$, {and the sampled covariates are pairwise distinct, as in Proposition~\ref{prop:equivalence_in_sample}}. {As in the main text, $\mathcal Y=[-1,1]$; the affineness axiom below averages model predictions within this interval.} For $f \in \bar{\mathcal{F}}$ we write $f|_S := (f(x_1), \dots, f(x_n)) \in [-1,1]^n$ for its vector of sample predictions. {Since the sampled covariates are pairwise distinct and $\bar{\mathcal F}$ contains all functions from $\mathcal X$ to $[-1,1]$, every vector in $[-1,1]^n$ is the restriction of some function to $S$.} Model classes are arbitrary nonempty subsets $\mathcal{F}_\Theta \subseteq \bar{\mathcal{F}}$.\footnote{No compactness of $\Theta$ or continuity in $\theta$ is needed in this appendix: all statements are made with infima, which the arguments never require to be attained. {Compactness and continuity of the discrepancy in parameters on the fixed sample suffice for attainment, but nothing here depends on it.}}

First, eligible sets are nonempty sets of synthetic outcome \textit{vectors} $F \subseteq \cube$ rather than sets of functions: an empirical object can distinguish functions only through their sample predictions, and the vectors $\sigma \in \cube$ are exactly the objects over which $\tilde{R}_n$ averages. {Under the maintained distinctness assumption, the vector and function formulations coincide by Lemma~\ref{lemma:sigma_function}.} We accordingly write $\tilde{e}_n(\mathcal{F}_\Theta, F, d, S)$ with a set of vectors in the second argument.

Second, {only within this appendix, }the empirical approximation error is treated as a \textit{primitive}: a functional
\[
\tilde{e}_n : (\mathcal{F}_\Theta,\, F,\, d,\, S) \longmapsto \tilde{e}_n(\mathcal{F}_\Theta, F, d, S) \in \mathbb{R}_+,
\]
defined for every nonempty model class $\mathcal{F}_\Theta \subseteq \bar{\mathcal{F}}$, every nonempty eligible set $F \subseteq \cube$, and every admissible discrepancy function $d$, rather than an object defined by the formula of Table \ref{tab:objects}. If the expected-infimum form and the uniform weighting were built into the definition, the aggregation axioms below would have nothing to determine; {the representation below identifies any remaining outcome-specific multipliers, which Symmetry equalizes at $d_n$ (Lemma \ref{lemma:finite_representation}). At this discrepancy, $\tilde{e}_n$ coincides with the object of Table \ref{tab:objects} up to a positive constant. The counting weights themselves are specified in Axiom \ref{ax:linearity}.}

Third, we use the following vector notation. For $\sigma \in \cube$, $\sigma^{(i)}$ denotes the vector that agrees with $\sigma$ except that the sign of its $i$th entry is reversed. For $\sigma, \sigma' \in \cube$ and $A \subseteq \{1, \dots, n\}$, $(\sigma_A, \sigma'_{-A})$ denotes the vector whose $i$th entry is $\sigma_i$ if $i \in A$ and $\sigma'_i$ otherwise. For $q \in \cube$ and $f \in \bar{\mathcal{F}}$, $q \odot f$ denotes (any function with) the sample predictions $(q_1 f(x_1), \dots, q_n f(x_n)) \in [-1,1]^n$, and $q \odot \sigma := (q_1 \sigma_1, \dots, q_n \sigma_n)$.

Finally, we specify the admissible discrepancy functions. Let $\Dn$ denote the set of functions $d : \bar{\mathcal{F}} \times \cube \rightarrow \mathbb{R}_+$ such that (i) $d(f, \sigma)$ depends on $f$ only through its sample predictions $f|_S$, and (ii) $d(f, \sigma) = 0$ if and only if $f|_S = \sigma$. Condition (i) is what makes a discrepancy function \textit{empirical}: it is computable from the sample. It also renders expressions such as $d(g, \sigma)$ well defined for any prediction vector $g \in [-1,1]^n$ --- evaluate $d$ at any $f \in \bar{\mathcal{F}}$ with $f|_S = g$, {whose existence follows from the distinctness of the sampled covariates and the definition of $\bar{\mathcal F}$} --- and we use this convention freely, e.g.\ $d(\sigma', \sigma)$ for $\sigma' \in \cube$. Condition (ii) is coincidence, the requirement placed on discrepancy functions by \cite{fudenberg2023restrictiveness}, read at the level of sample predictions, the only level at which an empirical object can satisfy it; as observed following Lemma \ref{lemma:inf_sup_correlation}, the discrepancy function used in Proposition \ref{prop:equivalence_in_sample} satisfies both conditions. We write $d_n$ for a generic element of $\Dn$; its characterization below recovers the discrepancy function of Lemma \ref{lemma:inf_sup_correlation} up to a positive multiple. For a model class we write, as in \cite{fudenberg2023restrictiveness}, $d(\mathcal{F}_\Theta, \sigma) := \inf_{f_\theta \in \mathcal{F}_\Theta} d(f_\theta, \sigma)$.

{Pairwise distinctness holds almost surely for iid samples from an atomless $P_X$. With repeated covariates, every function must assign identical predictions to identical covariates, whereas Rademacher signs remain independent across observations. A characterization covering that case can instead be formulated with discrepancy defined on prediction vectors and aggregation on arbitrary nonempty vector classes, then evaluated at the sample predictions of each actual model.}

\subsection{Axioms}\label{sec:axioms}

The first two axioms concern the eligible set $F \subseteq \cube$.

\begin{axiom}[Outcome Richness]\label{ax:outcome_richness}
    For every observation $i \in \{1, \dots, n\}$ and every $s \in \{-1, 1\}$, there exists $\sigma \in F$ such that $\sigma_i = s$.
\end{axiom}

\begin{axiom}[Combination]\label{ax:combination}
    For all $\sigma, \sigma' \in F$ and every $A \subseteq \{1, \dots, n\}$, $(\sigma_A, \sigma'_{-A}) \in F$.
\end{axiom}

Outcome richness says that neither synthetic outcome is ruled out at any observation. Combination says that eligibility imposes no restrictions linking synthetic outcomes across observations.

The next three axioms concern a fixed empirical discrepancy function $d_n \in \Dn$.

\begin{axiom}[Outcome Symmetry]\label{ax:outcome_symmetry}
    For every $q \in \cube$, $f \in \bar{\mathcal{F}}$, and $\sigma \in \cube$,
    \[
    d_n(q \odot f,\, q \odot \sigma) = d_n(f, \sigma).
    \]
\end{axiom}

\begin{axiom}[Equal Weighting]\label{ax:equal_weighting}
    For every $\sigma \in \cube$ and all observations $i, j \in \{1, \dots, n\}$,
    \[
    d_n(\sigma^{(i)}, \sigma) = d_n(\sigma^{(j)}, \sigma).
    \]
\end{axiom}

\begin{axiom}[Affineness]\label{ax:affineness}
    For every $f, g \in \bar{\mathcal{F}}$, $\sigma \in \cube$, and $\lambda \in [0, 1]$,
    \[
    d_n\big(\lambda f + (1-\lambda) g,\, \sigma\big) = \lambda\, d_n(f, \sigma) + (1-\lambda)\, d_n(g, \sigma).
    \]
\end{axiom}

Outcome symmetry says that the discrepancy function responds to whether predictions and synthetic outcomes agree, not to the names of the two outcomes. Equal weighting says that one full error has the same cost at every observation. Affineness says that the discrepancy of an averaged prediction is the average of the discrepancies; interpreting an interior prediction as the mean of a randomized binary prediction, it says that a randomized prediction is evaluated at its expected discrepancy. Affineness is a statement about the interval structure of $\mathcal{Y} = [-1,1]$; it has no counterpart when predictions are restricted to $\{-1,1\}$.

The remaining axioms concern the aggregation $\tilde{e}_n$. Throughout, $\mathcal{F}_\Theta, \mathcal{F}_{\Theta_1}, \mathcal{F}_{\Theta_2}$ are nonempty model classes, $F \subseteq \cube$ is a nonempty eligible set, $d \in \Dn$, and $\alpha > 0$. (The restriction to $\alpha > 0$ is deliberate: $\alpha = 0$ yields $d \equiv 0 \notin \Dn$.)

\begin{axiom}[Nonnegativity]\label{ax:nonnegativity}
    $\tilde{e}_n(\mathcal{F}_\Theta, F, d, S) \geq 0$.
\end{axiom}

\begin{axiom}[Monotonicity]\label{ax:monotonicity}
    If $d(\mathcal{F}_{\Theta_1}, \sigma) \geq d(\mathcal{F}_{\Theta_2}, \sigma)$ for all $\sigma \in F$, then
    \[
    \tilde{e}_n(\mathcal{F}_{\Theta_1}, F, d, S) \geq \tilde{e}_n(\mathcal{F}_{\Theta_2}, F, d, S).
    \]
\end{axiom}

\begin{axiom}[Homogeneity]\label{ax:homogeneity}
    (a) $\tilde{e}_n(\mathcal{F}_\Theta, F, \alpha d, S) = \alpha\, \tilde{e}_n(\mathcal{F}_\Theta, F, d, S)$.

    (b) If $d(\mathcal{F}_{\Theta_1}, \sigma) = \alpha\, d(\mathcal{F}_{\Theta_2}, \sigma)$ for all $\sigma \in F$, then
    \[
    \tilde{e}_n(\mathcal{F}_{\Theta_1}, F, d, S) = \tilde{e}_n(\mathcal{F}_{\Theta_2}, F, \alpha d, S).
    \]
\end{axiom}

\begin{axiom}[Linearity]\label{ax:linearity}
    For every partition of $F$ into nonempty disjoint sets $F_1, \dots, F_K$,
    \[
    \tilde{e}_n(\mathcal{F}_\Theta, F, d, S) = \sum_{k=1}^K \frac{|F_k|}{|F|}\, \tilde{e}_n(\mathcal{F}_\Theta, F_k, d, S).
    \]
\end{axiom}

\begin{axiom}[Symmetry]\label{ax:symmetry}
    For every bijection $\tau : F \rightarrow F$: if $d(\mathcal{F}_{\Theta_1}, \sigma) = d(\mathcal{F}_{\Theta_2}, \tau(\sigma))$ for all $\sigma \in F$, then $\tilde{e}_n(\mathcal{F}_{\Theta_1}, F, d, S) = \tilde{e}_n(\mathcal{F}_{\Theta_2}, F, d, S)$.
\end{axiom}

\begin{axiom}[Nondegeneracy]\label{ax:nondegeneracy}
    $\tilde{e}_n(\cdot, \cdot, d_n, S) \not\equiv 0$: some model class and eligible set satisfy $\tilde{e}_n(\mathcal{F}_\Theta, F, d_n, S) > 0$.
\end{axiom}

Axioms \ref{ax:nonnegativity}--\ref{ax:symmetry} are the axioms of Section 4 of \cite{fudenberg2023restrictiveness} with the prior over the eligible set taken to be counting measure on $\cube$, the canonical full-support choice in a finite environment; Axiom \ref{ax:linearity} is their Linearity with conditional weights $|F_k|/|F|$, and finite partitions suffice since $F$ is finite. Nondegeneracy is stated at the discrepancy function $d_n$ under consideration because Axioms \ref{ax:nonnegativity}--\ref{ax:symmetry} do not link the behavior of $\tilde{e}_n$ across non-proportional discrepancy functions: an aggregation could vanish identically at $d_n$ and not elsewhere.

\subsection{Characterization}\label{sec:characterization}

\begin{proposition}\label{prop:axiomatization}
    Fix a sample $S = (x_1, \dots, x_n)$ {of pairwise-distinct covariates} and let $\mathcal{Y} = [-1, 1]$. Suppose that
    \begin{enumerate}
        \item the eligible set $F \subseteq \cube$ is nonempty and satisfies Axioms \ref{ax:outcome_richness} and \ref{ax:combination};
        \item $d_n \in \Dn$ satisfies Axioms \ref{ax:outcome_symmetry}, \ref{ax:equal_weighting}, and \ref{ax:affineness};
        \item $\tilde{e}_n$ satisfies Axioms \ref{ax:nonnegativity}--\ref{ax:symmetry} and satisfies Axiom \ref{ax:nondegeneracy} at $d_n$.
    \end{enumerate}
    Then $F = \cube$; there exists a constant $a_n > 0$ such that
    \[
    d_n(f, \sigma) = a_n \left( 1 - \frac{1}{n} \sum_{i=1}^n f(x_i) \sigma_i \right) \qquad \text{for all } f \in \bar{\mathcal{F}},\ \sigma \in \cube;
    \]
    and there exists a constant $b_n > 0$ such that every nonempty model class $\mathcal{F}_\Theta \subseteq \bar{\mathcal{F}}$ satisfies
    \begin{equation}\label{eq:axiomatization}
    \tilde{e}_n(\mathcal{F}_\Theta, F, d_n, S) = b_n \left[ 1 - \tilde{R}_n(\mathcal{F}_\Theta, S) \right].
    \end{equation}
    Conversely, the eligible set $F = \cube$, every discrepancy function $d^*_\gamma(f, \sigma) := \gamma \left( 1 - \frac{1}{n} \sum_{i=1}^n f(x_i) \sigma_i \right)$ with $\gamma > 0$, and every functional
    \[
    \tilde{e}^*_n(\mathcal{F}_\Theta, F', d, S) := \kappa \cdot \frac{1}{|F'|} \sum_{\sigma \in F'} d(\mathcal{F}_\Theta, \sigma), \qquad \kappa > 0,
    \]
    satisfy conditions 1--3, and \eqref{eq:axiomatization} holds for them with $b_n = \kappa \gamma$.
\end{proposition}

{The constants may depend on the fixed sample and discrepancy, but not on the model class. The converse verifies the displayed construction; it does not require the same $\kappa$ at non-proportional discrepancies.}

The proof proceeds through four lemmas: Lemma \ref{lemma:eligible_set} characterizes the eligible set, Lemmas \ref{lemma:affine_representation} and \ref{lemma:discrepancy_characterization} characterize the discrepancy function, and Lemma \ref{lemma:finite_representation} characterizes the aggregation. The proof of the proposition then assembles the four.

\begin{lemma}\label{lemma:eligible_set}
    A nonempty $F \subseteq \cube$ satisfies Axioms \ref{ax:outcome_richness} and \ref{ax:combination} if and only if $F = \cube$.
\end{lemma}

\begin{proof}
    The full set clearly satisfies both axioms. Conversely, fix a target $s \in \cube$ and any $\sigma^0 \in F$. Axiom \ref{ax:outcome_richness} provides, for each $i$, a vector $\tau^i \in F$ with $\tau^i_i = s_i$. Recursively define $\sigma^i := (\tau^i_{\{i\}}, \sigma^{i-1}_{-\{i\}}) \in F$, applying Axiom \ref{ax:combination} with $A = \{i\}$. By induction the first $i$ entries of $\sigma^i$ agree with $s$; hence $\sigma^n = s \in F$.
\end{proof}

\begin{lemma}\label{lemma:affine_representation}
    Let $h : [-1,1]^n \rightarrow \mathbb{R}$ satisfy $h(\lambda v + (1-\lambda) w) = \lambda h(v) + (1-\lambda) h(w)$ for all $v, w \in [-1,1]^n$ and $\lambda \in [0,1]$. Then there are constants $\alpha, \beta_1, \dots, \beta_n$ such that
    \[
    h(v) = \alpha + \sum_{i=1}^n \beta_i v_i \qquad \text{for all } v \in [-1,1]^n.
    \]
\end{lemma}

\begin{proof}
    First, pairwise affineness extends to finite convex combinations by induction: for weights $\lambda_1, \dots, \lambda_m \geq 0$ summing to one with $\lambda_m < 1$, write $\sum_j \lambda_j v^j = (1-\lambda_m) u + \lambda_m v^m$ with $u := \sum_{j < m} \frac{\lambda_j}{1-\lambda_m} v^j \in [-1,1]^n$, and apply the two-point identity and the inductive hypothesis.

    Set $g(v) := h(v) - h(0)$, so that $g(0) = 0$. Then (i) $g(\lambda v) = g(\lambda v + (1-\lambda) 0) = \lambda g(v)$ for $\lambda \in [0,1]$, and (ii) $0 = g(0) = g\big(\frac{1}{2} v + \frac{1}{2}(-v)\big) = \frac{1}{2} g(v) + \frac{1}{2} g(-v)$, so $g(-v) = -g(v)$. Let $e^{(i)}$ denote the $i$th standard basis vector and $\beta_i := g(e^{(i)})$. For any $v$, the point $v/n$ is the convex combination with weights $1/n$ of the points $v_i e^{(i)} \in [-1,1]^n$, so by finite affineness and (i),
    \[
    \frac{1}{n}\, g(v) = g\left( \frac{v}{n} \right) = \frac{1}{n} \sum_{i=1}^n g\big(v_i e^{(i)}\big), \qquad \text{i.e.} \qquad g(v) = \sum_{i=1}^n g\big(v_i e^{(i)}\big).
    \]
    Finally, $g(v_i e^{(i)}) = v_i \beta_i$: for $v_i \geq 0$ this is (i) with $\lambda = v_i$; for $v_i < 0$, write $g(v_i e^{(i)}) = g\big(|v_i| (-e^{(i)})\big) = |v_i|\, g(-e^{(i)}) = -|v_i| \beta_i = v_i \beta_i$ using (i) and (ii). Hence $h(v) = h(0) + \sum_i \beta_i v_i$, and $\alpha := h(0)$.
\end{proof}

\begin{lemma}\label{lemma:discrepancy_characterization}
    A discrepancy function $d_n \in \Dn$ satisfies Axioms \ref{ax:outcome_symmetry}, \ref{ax:equal_weighting}, and \ref{ax:affineness} if and only if there exists $a_n > 0$ such that
    \begin{equation}\label{eq:corr_form}
    d_n(f, \sigma) = a_n \left( 1 - \frac{1}{n} \sum_{i=1}^n f(x_i) \sigma_i \right) \qquad \text{for all } f \in \bar{\mathcal{F}},\ \sigma \in \cube.
    \end{equation}
    In this case, $d_n(\sigma', \sigma) = d_n(\sigma, \sigma')$ for all $\sigma, \sigma' \in \cube$.
\end{lemma}

\begin{proof}
    \textit{Sufficiency.} Fix $a_n > 0$ and let $d_n$ be given by \eqref{eq:corr_form}, so that $d_n$ is $a_n$ times the discrepancy function of Lemma \ref{lemma:inf_sup_correlation}. It depends on $f$ only through $f|_S$, and, as observed following that lemma, it is nonnegative and vanishes if and only if $f|_S = \sigma$; both properties are preserved under multiplication by $a_n > 0$, so $d_n \in \Dn$. Axiom \ref{ax:outcome_symmetry} holds since $(q_i f(x_i))(q_i \sigma_i) = f(x_i) \sigma_i$; Axiom \ref{ax:equal_weighting} holds since $d_n(\sigma^{(i)}, \sigma) = \frac{2 a_n}{n}$ for every $i$; and Axiom \ref{ax:affineness} holds since \eqref{eq:corr_form} is affine in the predictions.

    \textit{Necessity.} Fix $\sigma \in \cube$ and regard $d_n(\cdot, \sigma)$ as a function of the prediction vector $v = f|_S \in [-1,1]^n$, which is well defined by property (i) of $\Dn$. If $f|_S = v$ and $g|_S = w$, then $(\lambda f + (1-\lambda) g)|_S = \lambda v + (1-\lambda) w$, so Axiom \ref{ax:affineness} states precisely that this function of $v$ is affine, and Lemma \ref{lemma:affine_representation} gives
    \begin{equation}\label{eq:affine_rep}
    d_n(v, \sigma) = \alpha(\sigma) + \sum_{i=1}^n \beta_i(\sigma)\, v_i
    \end{equation}
    for some coefficients $\alpha(\sigma), \beta_1(\sigma), \dots, \beta_n(\sigma)$.

    \textit{Axiom \ref{ax:outcome_symmetry} pins down the coefficients.} Fix $q \in \cube$. By Axiom \ref{ax:outcome_symmetry} and \eqref{eq:affine_rep}, for every $v \in [-1,1]^n$,
    \[
    \alpha(q \odot \sigma) + \sum_i \beta_i(q \odot \sigma)\, q_i v_i = d_n(q \odot v,\, q \odot \sigma) = d_n(v, \sigma) = \alpha(\sigma) + \sum_i \beta_i(\sigma)\, v_i.
    \]
    Evaluating at $v = 0$ gives $\alpha(q \odot \sigma) = \alpha(\sigma)$; since $\{ q \odot \sigma : q \in \cube \} = \cube$, the intercept is a constant $\alpha$. Evaluating at $v = e^{(i)}$, the $i$th standard basis vector, and subtracting the intercept gives $\beta_i(q \odot \sigma)\, q_i = \beta_i(\sigma)$. Taking $q = \sigma$, so that $q \odot \sigma$ is the all-ones vector $\mathbf{1}$, yields $\beta_i(\sigma) = \beta_i(\mathbf{1})\, \sigma_i =: w_i \sigma_i$. Hence
    \[
    d_n(v, \sigma) = \alpha + \sum_{i=1}^n w_i\, \sigma_i v_i.
    \]

    \textit{Coincidence pins down the intercept and the signs.} At $v = \sigma$, coincidence forces $0 = \alpha + \sum_i w_i$, so, writing $c_i := -w_i$,
    \[
    d_n(v, \sigma) = \sum_{i=1}^n c_i\, (1 - \sigma_i v_i).
    \]
    At $v = \sigma^{(i)}$ this equals $2 c_i$, which is nonnegative because $d_n$ maps into $\mathbb{R}_+$, and nonzero because $\sigma^{(i)} \neq \sigma$ and coincidence is an equivalence; hence $c_i > 0$ for every $i$.

    \textit{Axiom \ref{ax:equal_weighting} pins down the weights.} $d_n(\sigma^{(i)}, \sigma) = 2 c_i$ must be equal across $i$, so $c_i \equiv c > 0$ and
    \[
    d_n(v, \sigma) = c \sum_{i=1}^n (1 - \sigma_i v_i) = c n \left( 1 - \frac{1}{n} \sum_{i=1}^n v_i \sigma_i \right),
    \]
    which is \eqref{eq:corr_form} with $a_n := c n > 0$.

    Finally, under \eqref{eq:corr_form} both $d_n(\sigma', \sigma)$ and $d_n(\sigma, \sigma')$ equal $a_n \big( 1 - \frac{1}{n} \sum_i \sigma'_i \sigma_i \big)$, which is the symmetry statement.
\end{proof}

\begin{lemma}\label{lemma:finite_representation}
    Suppose $\tilde{e}_n$ satisfies Axioms \ref{ax:monotonicity}--\ref{ax:linearity}. Then there exists $c : \cube \times \Dn \rightarrow \mathbb{R}_+$ with $c(\sigma, \alpha d) = c(\sigma, d)$ for all $\alpha > 0$, such that for every $\mathcal{F}_\Theta$, $F$, and $d$,
    \begin{equation}\label{eq:finite_representation}
    \tilde{e}_n(\mathcal{F}_\Theta, F, d, S) = \frac{1}{|F|} \sum_{\sigma \in F} c(\sigma, d)\, d(\mathcal{F}_\Theta, \sigma).
    \end{equation}
    If $\tilde{e}_n$ also satisfies Axiom \ref{ax:symmetry}, then for every $d \in \Dn$ that is symmetric on binary pairs, i.e.\ $d(\sigma', \sigma) = d(\sigma, \sigma')$ for all $\sigma, \sigma' \in \cube$, there is $k(d) \geq 0$ with $c(\sigma, d) = k(d)$ for all $\sigma \in \cube$.
\end{lemma}

\begin{proof}
    \textit{Step 1: singletons.} Applying Axiom \ref{ax:linearity} to the partition of $F$ into singletons,
    \begin{equation}\label{eq:singletons}
    \tilde{e}_n(\mathcal{F}_\Theta, F, d, S) = \frac{1}{|F|} \sum_{\sigma \in F} \tilde{e}_n(\mathcal{F}_\Theta, \{\sigma\}, d, S).
    \end{equation}

    \textit{Step 2: linearity on singletons.} Fix $\sigma$ and $d$. If $d(\mathcal{F}_{\Theta_1}, \sigma) = d(\mathcal{F}_{\Theta_2}, \sigma)$, then Axiom \ref{ax:monotonicity} applied in both directions gives $\tilde{e}_n(\mathcal{F}_{\Theta_1}, \{\sigma\}, d, S) = \tilde{e}_n(\mathcal{F}_{\Theta_2}, \{\sigma\}, d, S)$; hence there is a nondecreasing function $\Phi = \Phi_{\sigma, d}$, defined on the set of achievable values $T_{\sigma, d} := \{ d(\mathcal{F}_\Theta, \sigma) : \mathcal{F}_\Theta \subseteq \bar{\mathcal{F}} \text{ nonempty} \}$, such that $\tilde{e}_n(\mathcal{F}_\Theta, \{\sigma\}, d, S) = \Phi\big(d(\mathcal{F}_\Theta, \sigma)\big)$. Note that $0 \in T_{\sigma, d}$: take $\mathcal{F}_\Theta = \{f\}$ with $f|_S = \sigma$. {A positive value is also achievable: take a singleton class with sample predictions $-\sigma$, giving $d(-\sigma,\sigma)>0$ by coincidence.}

    We first show $\Phi(0) = 0$. Take $\mathcal{F}_\Theta$ with $d(\mathcal{F}_\Theta, \sigma) = 0$ and any $\alpha \neq 1$. The pair $(\mathcal{F}_{\Theta_1}, \mathcal{F}_{\Theta_2}) = (\mathcal{F}_\Theta, \mathcal{F}_\Theta)$ satisfies the hypothesis of Axiom \ref{ax:homogeneity}(b), since both profiles are zero, so $\tilde{e}_n(\mathcal{F}_\Theta, \{\sigma\}, d, S) = \tilde{e}_n(\mathcal{F}_\Theta, \{\sigma\}, \alpha d, S)$, while Axiom \ref{ax:homogeneity}(a) gives $\tilde{e}_n(\mathcal{F}_\Theta, \{\sigma\}, \alpha d, S) = \alpha\, \tilde{e}_n(\mathcal{F}_\Theta, \{\sigma\}, d, S)$. Hence $\Phi(0) = \alpha \Phi(0)$, so $\Phi(0) = 0$.

    Next, linearity of $\Phi$. Let $t, t' \in T_{\sigma, d} \setminus \{0\}$ be realized by $\mathcal{F}_{\Theta_1}$ and $\mathcal{F}_{\Theta_2}$, and set $\alpha := t/t'$. By Axiom \ref{ax:homogeneity}(b) and then Axiom \ref{ax:homogeneity}(a),
    \[
    \Phi(t) = \tilde{e}_n(\mathcal{F}_{\Theta_1}, \{\sigma\}, d, S) = \tilde{e}_n(\mathcal{F}_{\Theta_2}, \{\sigma\}, \alpha d, S) = \alpha\, \tilde{e}_n(\mathcal{F}_{\Theta_2}, \{\sigma\}, d, S) = \frac{t}{t'}\, \Phi(t').
    \]
    So $\Phi(t)/t$ is constant on $T_{\sigma, d} \setminus \{0\}$; call the constant $c(\sigma, d)$, which is nonnegative since $\Phi$ is nondecreasing with $\Phi(0) = 0$. Then $\tilde{e}_n(\mathcal{F}_\Theta, \{\sigma\}, d, S) = c(\sigma, d)\, d(\mathcal{F}_\Theta, \sigma)$ for every $\mathcal{F}_\Theta$, and combining with \eqref{eq:singletons} gives \eqref{eq:finite_representation}. The invariance $c(\sigma, \alpha d) = c(\sigma, d)$ follows from Axiom \ref{ax:homogeneity}(a){: for any achievable $t>0$, that axiom gives $c(\sigma,\alpha d)\alpha t=\alpha c(\sigma,d)t$, and cancellation proves the claim.}

    \textit{Step 3: symmetry at a symmetric discrepancy.} Let $d$ be symmetric on binary pairs and fix $\sigma \neq \sigma'$. Apply Axiom \ref{ax:symmetry} with the eligible set $F = \{\sigma, \sigma'\}$, the transposition $\tau$ of $\sigma$ and $\sigma'$, and the singleton classes $\mathcal{F}_{\Theta_1} = \{f'\}$ and $\mathcal{F}_{\Theta_2} = \{f\}$, where $f'|_S = \sigma'$ and $f|_S = \sigma$. Their profiles satisfy
    \[
    d(\mathcal{F}_{\Theta_1}, \sigma) = d(\sigma', \sigma) = d(\sigma, \sigma') = d(\mathcal{F}_{\Theta_2}, \tau(\sigma)), \qquad d(\mathcal{F}_{\Theta_1}, \sigma') = 0 = d(\mathcal{F}_{\Theta_2}, \tau(\sigma')),
    \]
    by symmetry and coincidence respectively. Axiom \ref{ax:symmetry} and \eqref{eq:finite_representation} then give
    \[
    \tfrac{1}{2}\, c(\sigma, d)\, d(\sigma', \sigma) = \tfrac{1}{2}\, c(\sigma', d)\, d(\sigma, \sigma'),
    \]
    and $d(\sigma', \sigma) = d(\sigma, \sigma') > 0$ by coincidence, so $c(\sigma, d) = c(\sigma', d)$.
\end{proof}

\begin{proof}[Proof of Proposition \ref{prop:axiomatization}]
    That $F = \cube$ is Lemma \ref{lemma:eligible_set}, and the form of $d_n$ is Lemma \ref{lemma:discrepancy_characterization}. By Lemma \ref{lemma:finite_representation}, the representation \eqref{eq:finite_representation} holds; since $d_n$ is symmetric on binary pairs (Lemma \ref{lemma:discrepancy_characterization}), the second part of that lemma gives $c(\cdot, d_n) \equiv k(d_n) \geq 0$; and $k(d_n) = 0$ would force $\tilde{e}_n(\cdot, \cdot, d_n, S) \equiv 0$ by \eqref{eq:finite_representation}, contradicting Axiom \ref{ax:nondegeneracy}, so $k(d_n) > 0$. Now evaluate \eqref{eq:finite_representation} at $F = \cube$: the uniform average over $\cube$ is an expectation over a vector $\sigma$ of iid Rademacher random variables, so, setting $b_n := k(d_n)\, a_n > 0$,
    \begin{align*}
    \tilde{e}_n(\mathcal{F}_\Theta, F, d_n, S) & = k(d_n)\, \mathbb{E}_\sigma \left[ \inf_{f_\theta \in \mathcal{F}_\Theta} a_n \left( 1 - \frac{1}{n} \sum_{i=1}^n f_\theta(x_i) \sigma_i \right) \right] \\
    & = b_n\, \mathbb{E}_\sigma \left[ 1 - \sup_{f_\theta \in \mathcal{F}_\Theta} \frac{1}{n} \sum_{i=1}^n \sigma_i f_\theta(x_i) \right] = b_n \left[ 1 - \tilde{R}_n(\mathcal{F}_\Theta, S) \right],
    \end{align*}
    where the middle equality factors out $a_n$ and applies Lemma \ref{lemma:inf_sup_correlation}.

    For the converse, $F = \cube$ satisfies condition 1 by Lemma \ref{lemma:eligible_set}, and $d^*_\gamma$ satisfies condition 2 by the sufficiency half of Lemma \ref{lemma:discrepancy_characterization}. For $\tilde{e}^*_n$: Axiom \ref{ax:nonnegativity} is immediate; for Axiom \ref{ax:monotonicity}, pointwise dominance of profiles is preserved by averaging; for Axiom \ref{ax:homogeneity}(a), $(\alpha d)(\mathcal{F}_\Theta, \sigma) = \alpha\, d(\mathcal{F}_\Theta, \sigma)$; for Axiom \ref{ax:homogeneity}(b), replace $d(\mathcal{F}_{\Theta_1}, \sigma)$ by $\alpha\, d(\mathcal{F}_{\Theta_2}, \sigma) = (\alpha d)(\mathcal{F}_{\Theta_2}, \sigma)$ under the sum; for Axiom \ref{ax:linearity}, an average over $F$ regroups as the $|F_k|/|F|$-weighted mixture of the cell averages; for Axiom \ref{ax:symmetry}, reindex the sum by the bijection $\tau$; and for Axiom \ref{ax:nondegeneracy}, $\tilde{e}^*_n(\{f\}, \{\sigma\}, d^*_\gamma, S) = \kappa \gamma \big( 1 - \frac{1}{n} \sum_i f(x_i) \sigma_i \big) > 0$ whenever $f|_S \neq \sigma$. The computation above, with $\kappa$ in place of $k(d_n)$ and $\gamma$ in place of $a_n$, gives \eqref{eq:axiomatization} with $b_n = \kappa \gamma$.
\end{proof}

\newpage
\bibliographystyle{ecta}
\bibliography{references}

@inproceedings{liang2015learning,
  title={Learning with square loss: Localization through offset rademacher complexity},
  author={Liang, Tengyuan and Rakhlin, Alexander and Sridharan, Karthik},
  booktitle={Conference on Learning Theory},
  pages={1260--1285},
  year={2015},
  organization={PMLR}
}

@article{bartlett2002rademacher,
  title={Rademacher and Gaussian complexities: Risk bounds and structural results},
  author={Bartlett, Peter L and Mendelson, Shahar},
  journal={Journal of Machine Learning Research},
  volume={3},
  number={Nov},
  pages={463--482},
  year={2002}
}

@article{bartlett2002model,
  title={Model selection and error estimation},
  author={Bartlett, Peter L and Boucheron, St{\'e}phane and Lugosi, G{\'a}bor},
  journal={Machine Learning},
  volume={48},
  pages={85--113},
  year={2002},
  publisher={Springer}
}

@article{bartlett2005local,
  title={Local rademacher complexities},
  author={Bartlett, Peter L and Bousquet, Olivier and Mendelson, Shahar},
  year={2005},
    volume={33},
   ISSN={0090-5364},
   url={http://dx.doi.org/10.1214/009053605000000282},
   DOI={10.1214/009053605000000282},
   number={4},
   journal={The Annals of Statistics},
   publisher={Institute of Mathematical Statistics}
}

@inproceedings{lozano2000model,
  title={Model selection using Rademacher penalization},
  author={Lozano, Fernando},
  booktitle={Proceedings of the 2nd ICSC Symp. on Neural Computation (NC2000). Berlin, Germany. ICSC Academic Press},
  year={2000}
}

@inproceedings{beigman2006learning,
  title={Learning from revealed preference},
  author={Beigman, Eyal and Vohra, Rakesh},
  booktitle={Proceedings of the 7th ACM Conference on Electronic Commerce},
  pages={36--42},
  year={2006},
  publisher={Association for Computing Machinery},
  doi={10.1145/1134707.1134712}
}

@article{farrell2021deep,
  title={Deep neural networks for estimation and inference},
  author={Farrell, Max H and Liang, Tengyuan and Misra, Sanjog},
  journal={Econometrica},
  volume={89},
  number={1},
  pages={181--213},
  year={2021},
  publisher={Wiley Online Library}
}

@book{hansen2022probability,
  title={Probability and statistics for economists},
  author={Hansen, Bruce},
  year={2022},
  publisher={Princeton University Press}
}

@article{fudenberg2022measuring,
  title={Measuring the completeness of economic models},
  author={Fudenberg, Drew and Kleinberg, Jon and Liang, Annie and Mullainathan, Sendhil},
  journal={Journal of Political Economy},
  volume={130},
  number={4},
  pages={956--990},
  year={2022},
  publisher={The University of Chicago Press Chicago, IL}
}

@article{gentzkow2024pricing,
  title={Pricing power in advertising markets: Theory and evidence},
  author={Gentzkow, Matthew and Shapiro, Jesse M and Yang, Frank and Yurukoglu, Ali},
  journal={American Economic Review},
  volume={114},
  number={2},
  pages={500--533},
  year={2024},
  publisher={American Economic Association 2014 Broadway, Suite 305, Nashville, TN 37203}
}

@Article{fudenberg2023restrictiveness,
  Title={How Flexible Is That Functional Form? Quantifying the Restrictiveness of Theories},
  Author={Fudenberg, Drew and Gao, Wayne and Liang, Annie},
  journal={Review of Economics and Statistics},
  volume={108},
  number={1},
  pages={194--209},
  year={2026},
  doi={10.1162/rest_a_01401}
}

@book{mitzenmacher2017probability,
  title={Probability and computing: Randomization and probabilistic techniques},
  author={Mitzenmacher, Michael and Upfal, Eli},
  year={2017},
  publisher={Cambridge university press}
}

@article{basu2020falsifiability,
  title={On the falsifiability and learnability of decision theories},
  author={Basu, Pathikrit and Echenique, Federico},
  journal={Theoretical Economics},
  volume={15},
  number={4},
  pages={1279--1305},
  year={2020},
  publisher={Wiley Online Library}
}

@article{schmeidler1989subjective,
  title={Subjective probability and expected utility without additivity},
  author={Schmeidler, David},
  journal={Econometrica: Journal of the Econometric Society},
  pages={571--587},
  year={1989},
  publisher={JSTOR}
}

@article{fudenberg2026model,
  title={Model Restrictiveness in Functional and Structural Settings},
  author={Fudenberg, Drew and Gao, Wayne Yuan and You, Zhiheng},
  journal={arXiv preprint arXiv:2602.07688},
  year={2026}
}

@book{koenker2005quantile,
  title={Quantile Regression},
  author={Koenker, Roger},
  series={Econometric Society Monographs},
  number={38},
  publisher={Cambridge University Press},
  year={2005},
  doi={10.1017/CBO9780511754098}
}

@article{liang2026using,
  author = {Annie Liang},
  title = {Using Machine Learning to Generate, Clarify, and Improve Economic Models},
  journal = {Journal of Economic Literature},
  note = {Accepted},
  year = {2026}
}

\end{document}